\documentclass[11pt]{amsart}
\usepackage{amsmath,amsthm,amssymb}
\usepackage{relsize}
\usepackage[noadjust]{cite}
\usepackage{color}
\usepackage{graphicx}
\usepackage[dvipsnames]{xcolor}
\usepackage{multirow}
\usepackage{booktabs}
\usepackage{tikz}
\usetikzlibrary{decorations.pathreplacing}
\usepackage{pgfplots}
\pgfplotsset{compat=1.18}

\usepackage[colorlinks, urlcolor=blue, linkcolor=Purple, citecolor=red]{hyperref}
\usepackage[capitalize]{cleveref}

\usepackage{setspace}
\usepackage{enumitem}
\setitemize{itemsep=-1pt}
\setenumerate{itemsep=-1pt}

\usepackage[margin=2.4cm]{geometry}

\usepackage{stmaryrd}
\usepackage{longtable}

\usepackage{array}
\newcolumntype{x}[1]{>{\centering\arraybackslash\hspace{0pt}}p{#1}}

\theoremstyle{definition}
\newtheorem{theorem}{Theorem}[section]
\newtheorem{definition}[theorem]{{{Definition}}}
\newtheorem{example}[theorem]{{{Example}}}

\newtheorem{notation}[theorem]{{{Notation}}}
\newtheorem{remark}[theorem]{{{Remark}}}

\newtheorem{corollary}[theorem]{{{Corollary}}}%[theorem]
\newtheorem{proposition}[theorem]{{{Proposition}}}
\newtheorem{lemma}[theorem]{{{Lemma}}}

\newcommand{\F}{\mathbb{F}}

\newcommand{\C}{\mathcal{C}}
\newcommand{\mC}{\mathcal{C}}

\newcommand{\M}{\mathcal{M}}

\newcommand{\sH}{\supp}

\newcommand{\wt}{\textnormal{wt}}
\newcommand{\Fq}{\F_q}

\newcommand{\colsp}{\textnormal{colsp}}

\newcommand{\Fqh}{\mathbb{F}_{q^h}}
\newcommand{\Fqhn}{\mathbb{F}_{q^h}^n}

\DeclareMathOperator{\supp}{supp}

\DeclareMathOperator{\PG}{PG}

\DeclareMathOperator{\rowsp}{rowsp}

\DeclareMathOperator{\dd}{d}

\newcommand{\cC}{\mathcal{C}}

\title{Minimal additive codes and additive strong blocking sets}

\usepackage[foot]{amsaddr}
\author{Gianira N. Alfarano}
\author{Marine Le Meur}
\address{Universit\'e de Rennes, IRMAR, Campus de Beaulieu, F-35042 Rennes Cedex, France.}
\email{gianira-nicoletta.alfarano@univ-rennes.fr, marine.le-meur@univ-rennes.fr}

\begin{document}

\begin{abstract}
  Additive codes over $\F_{q^h}$ generalize linear codes by relaxing linearity over the alphabet while retaining linearity over the subfield $\Fq$. In this paper, we introduce minimal additive codes and we initiate their study from a geometric perspective. We define the concept of additive strong blocking sets, a class of $h$-projective systems whose union forms a strong blocking set. We establish a one-to-one correspondence between equivalence classes of nondegenerate minimal additive codes and equivalence classes of additive strong blocking sets. We also compare this framework with the theory of outer strong blocking sets, showing that the latter arises as a special case. Finally, we provide constructions and existence results for minimal additive codes, and derive upper, lower, and asymptotic bounds on their minimum length.
\end{abstract}

\maketitle



\tableofcontents

\medskip

\section{Introduction}
A linear code $\mC\subseteq \F_q^n$ is a $k$-dimensional subspace of the $n$-dimensional vector space over the finite field $\Fq$, endowed with the Hamming distance. The
minimum distance $d$ of $\mC$ is the least distance between two distinct codewords. It is well-known that a nondegenerate linear code can be described
geometrically by a \emph{projective system}, that is, a multiset of $n$ points in the projective space $\PG(k-1,q)$, such that each hyperplane contains at most $n-d$ of them. This correspondence is one of the classical links between coding theory and
finite geometry: metric properties of the code are translated into incidence properties of points and hyperplanes; see for instance \cite{tsfasman2007algebraic}. 

Additive codes provide a natural generalization of linear codes. An additive code of length $n$ over~$\Fqh$ is an $\Fq$-subspace of $\Fqh^n$.
Thus, an additive code need not be linear over the alphabet $\Fqh$, but it still has an $\Fq$-linear structure. This weaker linearity is reflected in the
geometric viewpoint. Instead of points in $\PG(k-1,q)$, additive codes correspond
to collections of subspaces of $\PG(k-1,q)$ of dimension at most $h-1$ and the minimum distance is determined by the maximum number of these subspaces contained in a hyperplane. This
geometric approach has recently led to several results on the parameters of additive codes and to the discovery of new additive codes whose parameters outperform those of linear codes; see e.g.
\cite{kurz2024additive,d2026generalized,adriaensen2023additive,bartoli2025long}. In these works, particular attention has been given to \emph{faithful} additive codes, where the dimension of each associated space is exactly $h-1$.
Additive codes also arise naturally in the construction of quantum codes and this further motivated their investigation; see
\cite{dastbasteh2024new,dastbasteh2025polynomial,grassl2021algebraic}.

Minimal linear codes are another class of codes where the interaction between coding theory and finite geometry is particularly strong. They were originally motivated by applications to secret sharing schemes \cite{massey1993minimal},
and have since been studied for their combinatorial and geometric structure. In \cite{alfarano2022geometric} and \cite{tang2021full}, it was independently
proved that nondegenerate minimal linear codes are in one-to-one correspondence with projective systems whose points form a strong blocking set. This characterization has become a central tool in the study of minimal codes, leading to constructions, lower bounds, and asymptotic results; see for
example
\cite{alfarano2022three,heger2021short,bartoli2023small,alon2024strong, bishnoi2024blocking}.

In \cite{alfarano2024outer},  the  notions  of  \emph{outer  strong  blocking  sets}  and  \emph{outer  minimal  codes} have been introduced.  These  are  sets of points whose  field  reduction  is  a  strong  blocking  set  and  linear codes whose concatenation with minimal linear codes is minimal. We generalize these notions by weakening the field of linearity. 

In this paper, we first prove a total-weight formula
for additive codes in terms of the vectorial dimensions of the associated subspaces; see Theorem~\ref{thm:tot_weight}. This leads to a refined Plotkin-type bound which makes use of these dimensions (Theorem~\ref{thm:plotkin_ti}), and to an asymptotic version (Theorem~\ref{thm:asymptotic_plotkin_ti}). These results show that non-faithful additive codes may behave differently from codes obtained from a fixed alphabet of size $q^h$.
We also apply the total-weight formula to one-weight additive codes and obtain examples of codes which are not equivalent to any $\Fqh$-linear code; see Theorem~\ref{thm:simplex_add}.

We introduce the more general concepts of \emph{minimal additive codes} and \emph{additive strong blocking sets}.
We prove that equivalence classes of nondegenerate minimal additive codes are in one-to-one
correspondence with equivalence classes of additive strong blocking sets; see Theorem~\ref{thm:minimal_additive_sbs}. This
extends the classical correspondence between minimal linear codes and strong blocking sets, and it generalizes the theory of outer minimal codes and outer strong blocking sets; see Section~\ref{sec:Linear_case}. In this sense, additive strong blocking sets provide a common geometric framework for both minimal linear codes and
outer minimal codes.

After establishing the correspondence, we study the parameters of minimal additive codes. We introduce the parameter $m(k,q,h)$, the smallest length of a nondegenerate minimal additive code of
$\Fq$-dimension $k$ over $\Fqh$. We prove an existence result (Theorem~\ref{thm:existence}) and give geometric constructions from additive strong blocking sets (Proposition~\ref{prop:upper-subspaces}). Inspired by the theory of strong blocking sets, we also derive
lower bounds using the Jamison--Brouwer--Schrijver (\cite{jamison1977covering, brouwer1978blocking}) theorem for affine blocking sets, including a
lower bound on $m(k,q,h)$  in Theorem~\ref{thm:lower-geometric} and a lower bound on the minimum distance in Theorem~\ref{thm:distance-lower-bound}. Finally, combining the minimality constraint with the Plotkin-type bound, we obtain asymptotic lower bounds for $m(k,q,h)$ in Theorem~\ref{thm:asympt_rate} and show that the bound in Theorem~\ref{thm:lower-geometric} is not asymptotically tight for $h\geq 2$, leaving space for improvements.

The paper is organized as follows. In Section \ref{sec:background}, we provide the necessary background to the rest of the paper. In Section~\ref{sec:bounds_on_additive} we derive a total-weight bound and a Plotkin-type bound on additive codes and apply them to one-weight codes. In Section \ref{sec:min_codes}, we introduce the main concepts of minimal additive codes and additive strong blocking sets. We initiate their study, providing a one-to-one correspondence between them.  In Section \ref{sec:constructions_bounds}, we provide a first existence result and several bounds on the parameters of these codes. We conclude in Section \ref{sec:conclusion} with some future research directions.

\medskip

\section*{Acknowledgments}
This research is supported by the grants ANR-21-CE39-0009-BARRACUDA and ANR-24-CPJ1-0075-01. The authors are thankful to Martino Borello, Delphine Boucher and Alessandro Neri for fruitful discussions.

\medskip

\section{Background}\label{sec:background}

In this section, we recall some basic notions on linear and additive codes, and we fix the notation for the rest of the paper. We first review the classical correspondence between linear codes and projective systems, and then we recall its extension to additive codes. For more details on additive codes and their geometric interpretation, we refer to \cite{ball2020additive, kurz2024additive, ball2025additive, bartoli2025long}.

\begin{notation}\label{not:initial}
Let $q$ be a prime power, $\F_q$ be the finite field of order $q$ and $\F_{q^h}$ be the extension of $\F_q$ of degree $h$. The space $\F_{q^h}^n$ is the $n$-dimensional vector space over $\F_{q^h}$.  All vectors are denoted as row vectors, unless otherwise stated. We write $[n]:=\{1,\ldots,n\}$. Let $\mathcal{B}= \{\alpha_1,\ldots,\alpha_h\}$ be a basis of $\F_{q^h}$ over $\mathbb{F}_q$ and denote $\boldsymbol{\alpha}=(\alpha_1,\ldots, \alpha_h)^\top$. We denote by $\mathcal{M}_{r,s}(\F_q)$ the space of $r\times s$ matrices with entries in $\Fq$. 
For every $0\leq\ell\leq n$, the number of $\ell$-dimensional subspaces of $\F_q^n$ is denoted by the Gaussian binomial coefficient
$$\binom{n}{\ell}_q=\frac{(q^n-1)(q^{n-1}-1)\cdots(q^{n-\ell+1}-1)}{(q-1)(q^2-1)\cdots(q^\ell-1)}.$$
Given a vector space $V\subseteq\F_q^k$, we denote by $\PG(V)$ the corresponding projective space. In particular, we denote by $\PG(k-1,q)$ the projective space with underlying vector space $\F_q^k$.
\end{notation}

Let $\F$ be any finite field. Given a vector $v\in\F^n$, its \textbf{(Hamming) support} is the set $\supp(v):=\{i\in[n]:v_i\neq 0\}$, and its \textbf{Hamming weight} is $\wt(v)=|\supp(v)|$. The Hamming weight induces the \textbf{Hamming distance} on $\F^n$, which is defined as $\dd(u,v)=\wt(u-v)$, for every $u,v\in\F^n$. In the rest of the paper we will consider vectors in $\F_q^n$ and in $\F_{q^h}^n$. In both cases, the notions of support and distance do not change. 

\subsection{Linear codes and projective systems}

An $[n,k]_q$ \textbf{linear code} $\mC$ is a $k$-dimensional $\F_q$-linear subspace of $\F_q^n$. The \textbf{minimum (Hamming) distance} of $\mC$ is defined as
$$\dd(\mC) := \min\{\dd(u,v) : u,v\in\mC, \; u\ne v\} = \min \{\wt(c) : c\in\mC\setminus\{0\}\}.$$
If $d=\dd(\mC)$ we say that $\mC$ is an $[n,k,d]_q$ linear code.
A \textbf{generator matrix} of $\mC$ is a matrix $G\in\mathcal{M}_{k,n}(\F_q)$ whose rows form an $\F_q$-basis of $\mC$. We say that $\mC$ is \textbf{nondegenerate} if there is no coordinate $i\in[n]$ such that $c_i=0$ for all $c\in\mC$.

Let $\mC$ be a nondegenerate $[n,k,d]_q$ linear code and let $G\in\mathcal{M}_{k,n}(\F_q)$ be a generator matrix for $\mC$. Writing 
$$G=\left(
\begin{array}{c|c|c}
g_1 & \cdots & g_n
\end{array}
\right),$$
where $g_i\in\F_q^k$ is the $i$-th column of $G$, each column defines a point $P_i$ of $\PG(k-1,q)$, since $\mC$ is nondegenerate. The \textbf{projective system associated with $\mC$} is the multiset of points $\mathcal{P}_G(\mC)=\{P_1,\ldots,P_n\}\subseteq \PG(k-1,q)$. Up to projective equivalence, this multiset does not depend on the choice of $G$, and we simply denote it by $\mathcal{P}(\mC)$.
It is well-known that  the minimum Hamming distance of $\mC$ can be recovered from $\mathcal{P}(\mC)$. Indeed, if $c=aG$ is a nonzero codeword, with $a\in\F_q^k\setminus\{0\}$, then the zero coordinates of $c$ correspond exactly to the points $P_i$ contained in the hyperplane $\langle a\rangle_{\Fq}^\perp$. Hence
$$\wt(c)=n-|\{i\in[n]:P_i\in \langle a\rangle_{\Fq}^\perp\}|,$$
and therefore
$$d=n-\max_{\substack{H\subseteq \PG(k-1,q),\\
H\textnormal{ hyperplane}}}
|\{i\in[n]:P_i\in H\}|.$$
Vice versa, given a set of points $\mathcal{S}\subseteq \PG(k-1,q)$, we construct the generator matrix $G_\mathcal{S}$ of a linear code  with columns given by the coordinates of the points in $\mathcal{S}$.
Thus, equivalence classes of nondegenerate $[n,k,d]_q$ linear codes are in one-to-one correspondence with equivalence classes of multisets of $n$ points in $\PG(k-1,q)$, not contained in a hyperplane, such that every hyperplane contains at most $n-d$ of them and some hyperplane contains exactly $n-d$ of them; see~\cite{tsfasman2007algebraic} for more details.

\subsection{Additive codes and $h$-projective systems} We explain here the geometric viewpoint on additive codes, that generalizes the concept of projective system. We first start with some basic definitions.

An $[n, k/h]_q^h$ \textbf{additive code} $\cC$ is a $k$-dimensional $\F_q$-linear subspace of~$\F_{q^h}^n$.  This standard notation highlights the formal ratio $k/h$, which may or may not be an integer, and it is equal to $\log_{q^h}|\C|$. 
Since the code is additive, we can define the \textbf{minimum (Hamming) distance} of $\cC$ as usual, namely
$$\dd(\cC) := \min \{\dd(u,v) : u,v \in\cC, \; u\ne v\} =\min \{\wt(c) : c\in\mC\setminus\{0\}\} .$$
When $d=\dd(\mC)$ is known, we say that $\cC$ is an $[n,k/h,d]_q^h$ additive code.
A matrix $G\in\mathcal{M}_{k,n}(\F_{q^h})$ is a \textbf{generator matrix} for $\mathcal{C} \subseteq \Fqhn$ if its rows are $\F_q$-linearly independent and span~$\mC$ over $\F_q$.  We say that $\mC$ is \textbf{nondegenerate} if there is no coordinate $i\in[n]$ such that $c_i=0$ for all $c\in\mC$.

\begin{definition} Two $[n, k/h,d]_q^h$ additive codes $\mC,\mathcal{D}$ are \textbf{equivalent} if one can be obtained from the other by permuting the coordinate positions, or by applying an $\Fq$-linear automorphism of $\Fqh$ in each coordinate position.
\end{definition}

Recently, additive codes have been studied in connection to finite geometry in  \cite{ball2025griesmer}. Since then, many other works used this approach to investigate several interesting constructions and bounds on the parameters; see e.g. \cite{ball2020additive, kurz2024additive, bartoli2025long, adriaensen2023additive, d2026generalized}.
Let $\mC$ be a $[n,k/h,d]_q^h$ code and let $G\in \M_{k,n}(\F_{q^h})$ be a generator matrix for $\mC$. We write  $$G=\left(
\begin{array}{c|c|c|c}
\textcolor{red}{g_1} & \textcolor{ForestGreen}{g_2} & \cdots & \textcolor{blue}{g_n}
\end{array}
\right),$$ 
where $g_i=(g_{i,1}, \ldots, g_{i,k})^\top\in\F_{q^h}^k$ is the $i$-th column of $G$.
Each $g_i$ is a vector with entries in $\F_{q^h}$, hence we can expand it over $\F_q$, with respect to the $\F_q$-basis~$\mathcal{B}$, i.e. $g_{i,j}=\sum_{t=1}^h b_{i,j}^t\alpha_t$, where $b_{i,j}^{t}\in\Fq$ for every $i,j,t$. 
This yields a matrix $\widetilde{G}\in \M_{k,nh}(\mathbb{F}_q)$ with entries in $\mathbb{F}_q$, whose
columns are grouped into~$n$ blocks of size $h$:

\begin{equation}\label{eq:Gtilde}
\widetilde{G} =
\begin{pmatrix}
\textcolor{red}{b^{1}_{1,1}} & \cdots & \textcolor{red}{b^{h}_{1,1}} &
\textcolor{ForestGreen}{b^{1}_{1,2}} & \cdots & \textcolor{ForestGreen}{b^{h}_{1,2}} & \cdots &
\textcolor{blue}{b^{1}_{1,n}} & \cdots & \textcolor{blue}{b^{h}_{1,n}} \\
\textcolor{red}{b^{1}_{2,1}} & \cdots & \textcolor{red}{b^{h}_{2,1}} &
\textcolor{ForestGreen}{b^{1}_{2,2}} & \cdots & \textcolor{ForestGreen}{b^{h}_{2,2}} & \cdots &
\textcolor{blue}{b^{1}_{2,n}} & \cdots & \textcolor{blue}{b^{h}_{2,n}} \\
\vdots & & \vdots & \vdots & & \vdots & & \vdots & & \vdots \\
\textcolor{red}{b^{1}_{k,1}} & \cdots & \textcolor{red}{b^{h}_{k,1}} &
\textcolor{ForestGreen}{b^{1}_{k,2}} & \cdots & \textcolor{ForestGreen}{b^{h}_{k,2}} & \cdots &
\textcolor{blue}{b^{1}_{k,n}} & \cdots & \textcolor{blue}{b^{h}_{k,n}}
\end{pmatrix}
\end{equation}

If $\mC$ is nondegenerate,  the columns of each block of $\widetilde{G}$ span a non-zero $\F_q$-subspace of $\F_q^k$ of dimension at most $h$.
We can regard each of these spaces as a projective subspace of $\PG(k-1,q)$ of dimension at most $h-1$. Let $\chi_G(\mC)=\{\pi_1,\ldots,\pi_n\}\subseteq \PG(k-1,q)$ be the multiset of such projective subspaces. In \cite{ball2020additive, ball2025additive}, the link between $\mC$ and $\chi_G(\mC)$ has been investigated, and, in particular, it is shown that the choice of $G$ does not change the multiset $\chi_G(\mC)$,
so we can simply use the notation $\chi(\mC)$, without specifying $G$. We recall the following formal definition from \cite{ball2025griesmer}.

\begin{definition}
     Let $\chi=\{\pi_1,...,\pi_n\}$ be a multiset of subspaces of $\PG(k-1,q)$. We say that $\chi$ is a $h$-$[n,k,d]_q$  \textbf{projective system} if $\dim(\pi_i)\leq h-1$ for every $i\in[n]$ and the spaces $\pi_i$ are not contained in the same hyperplane. Moreover, we have that
     $$n-d=\max_{\substack{H\subseteq \PG(k-1,q),\\
     H\textnormal{ hyperplane }}}|\{i\in [n]: \pi_i \subseteq H\}|,$$
     i.e. every hyperplane contains at most $n-d$ elements of $\chi$ and there exists a hyperplane that contains exactly $n-d$ elements of $\chi$. When the parameters are clear or not relevant we simply refer to $\chi$ as \textbf{$h$-projective system}. 
\end{definition}

\begin{definition}
    We say that two $h$-[$n,k,d]_q$ projective systems $\chi$ and $\psi$ are \textbf{equivalent} if there exists a $\varphi \in \text{PGL}_k(\Fq)$ such that $\varphi(\chi) = \psi$.
\end{definition} 

The following is the generalization to the additive case of the well-known correspondence between linear codes and projective systems; see~\cite{tsfasman2007algebraic}. 

\begin{theorem}[{\cite[Theorem 5]{ball2025griesmer}}]
    There is a one-to-one correspondence between equivalence classes of nondegenerate $[n,k/h,d]_q^h$ codes and equivalence classes of $h$-$[n,k,d]_q$ projective systems.
\end{theorem}

We briefly explain the correspondence. Let $\mathcal{C}$ be an $[n,k/h,d]_q^h$ additive code and $\chi(\mC)$ be defined as above. Then $\chi(\mC)$ is an $h$-$[n,k,d]_q$ projective system. Indeed, if $c \in \mathcal{C}\setminus\{0\}$ and $a \in \mathbb{F}_q^k\setminus\{0\}$ such that $c=aG$, the weight of $c$ can be obtained as
$$\wt(c)=\wt(aG)=n-|\{i\in[n] :  \pi_i \subseteq \langle a\rangle_{\F_q}^\perp\}|.$$
Therefore, the minimum distance of the code $\mathcal{C}$ is
$$d=n-\max_{\substack{H\subseteq \PG(k-1,q),\\
     H\textnormal{ hyperplane }}}|\{i\in [n]: \pi_i \subseteq H\}|.$$
In particular, every nonzero codeword $c=aG$ corresponds to the hyperplane $\langle a \rangle_{\F_q}^\perp \subseteq \PG(k-1,q)$. 

Vice versa, let $\{\pi_1,\ldots,\pi_n\}$ be an $h$-$[n,k,d]_q$ projective system.
For each~$\pi_i$, choose a $k\times h$ matrix~$G_i$ such that $\colsp(G_i)=\pi_i$. Construct $G =(G_1\boldsymbol{\alpha} \; | \; \cdots \; | G_n\boldsymbol{\alpha}  )\in\mathcal{M}_{k,n}(\F_{q^h})$, where $\boldsymbol{\alpha}$ is as in Notation \ref{not:initial}. Then $\mC=\rowsp_{\F_q}(G)$ is an $[n,k/h,d]_q^h$ additive code. The correspondence described above depends on the choices of generator matrices for $\mathcal{C}$ or bases for the $\pi_i$'s. However, it is well-defined at the level of equivalence classes: different choices yield equivalent codes and projective systems. Given a nondegenerate $[n,k/h,d]_q^h$ additive code $\mC$, we refer to $\chi(\mC)$ as \textbf{an $h$-projective system associated with $\mC$}.

%%%%%%%%%%%%%%%%%%%%%%%%%%%%%%%%%%%%%%%%%%%%%%%%%%%%55
\section{Bounds on the parameters of additive codes}\label{sec:bounds_on_additive}
In this section, we derive some elementary bounds on the parameters of additive codes. We first compute the total weight of a nondegenerate additive code in terms of the dimensions of the associated projective subspaces. We then obtain a refined Plotkin-type bound in the additive setting and
deduce its asymptotic form. As an application, we obtain conditions on the parameters of one-weight additive codes.

\subsection{Total weight bound}

Let $\mC$ be a nondegenerate $[n,k/h,d]_q^h$ additive code and let $\chi(\mC)=\{\pi_1,\ldots,\pi_n\}$ be an associated $h$-projective system in $\PG(k-1,q)$. For each $i\in[n]$, let $t_i$ denote the vectorial dimension of the subspace corresponding to $\pi_i$. Thus, the projective dimension of $\pi_i$ is $t_i-1$, and $1\le t_i\le \min\{h,k\}$.

 We first recall the definition of \emph{faithful} additive codes from \cite{ball2025griesmer}. This notion is particularly relevant and has been used for deriving bounds and constructions of additive codes; see e.g. \cite{kurz2024additive}.

 \begin{definition}
A nondegenerate additive code $\C$ is said to be \textbf{faithful} if $t_i=h$ for every $i\in[n]$, or equivalently, if every $\pi_i$ has projective dimension $h-1$.
\end{definition}

We obtain the following \emph{total weight bound} for a nondegenerate additive code.

\begin{theorem}\label{thm:tot_weight}
   Let $\mC$ be a nondegenerate $[n,k/h,d]_q^h$ additive code and let $\chi(\mC)=\{\pi_1,\ldots,\pi_n\}$ be an~$h$-projective system associated with $\mC$. Assume that $\pi_i$ has vectorial dimension $t_i$, for every $i\in[n]$. 
   Then, $$\sum_{c \in \mC} \wt(c)=n(q^k-1)-\sum_{i=1}^n (q^{k-t_i}-1).$$
\end{theorem}
\begin{proof}
Let $G$ be the generator matrix of $\C$ obtained by $\chi(\C)$. For every nonzero codeword $c\in \C$, there exists $u\in \F_q^k\setminus\{0\}$ such that $c=uG$. Moreover, the $i$-th coordinate of $c$ is zero if and only if $\pi_i\subseteq \langle u\rangle_{\Fq}^\perp$.
Therefore, the total weight of the code $\C$ verifies
     \begin{align*}
    \sum_{c \in \mC}\wt(c)=& \sum_{i=1}^n|\{c \in \mathcal{C} \; : \; c_i\neq 0\}| 
    =\sum_{i=1}^n(|\mathcal{C}|- |\{c \in \mathcal{C} \; : \; c_i=0\}|) 
    =&\sum_{i=1}^n \left( q^k- |\{c \in \mathcal{C} \; : \; c_i=0\}|\right).
\end{align*} 
We now compute $|\{c\in \C : c_i=0\}|$. Since the zero codeword has zero $i$-th coordinate, we have
$$ |\{c\in \C : c_i=0\}|=1+|\{u\in \F_q^k\setminus\{0\} \; : \; \pi_i\subseteq \langle u\rangle_{\Fq}^\perp\}|.
$$
The hyperplanes of $\PG(k-1,q)$ containing $\pi_i$ are exactly
$$
        \binom{k-t_i}{1}_q =\frac{q^{k-t_i}-1}{q-1}.
$$
Each such hyperplane is represented by $q-1$ nonzero vectors $u\in\F_q^k$, up to multiplication by elements of
$\F_q^*$. Hence,
$$
|\{u\in \F_q^k\setminus\{0\} \; : \;\pi_i\subseteq \langle u\rangle_{\Fq}^\perp\}|=(q-1)\binom{k-t_i}{1}_q =q^{k-t_i}-1.
$$
Thus, we have
$$
        |\{c\in \C : c_i=0\}|=q^{k-t_i},
$$
and, hence,
$$
        |\{c\in \C : c_i\neq 0\}|=q^k-q^{k-t_i}.
$$
Summing over $i=1,\ldots,n$, we obtain
$$
        \sum_{c\in \C}\wt(c) =\sum_{i=1}^n(q^k-q^{k-t_i})
        = n(q^k-1)-\sum_{i=1}^n(q^{k-t_i}-1).
$$
\end{proof}

As a consequence, we obtain the following bound on the parameters of an additive code.

\begin{corollary}
Let $\C$ be a nondegenerate $[n,k/h,d]^h_q$ additive code and let $\chi(\C)=\{\pi_1,\ldots,\pi_n\}$ be an $h$-projective system associated with $\mC$. Assume that each $\pi_i$ has vectorial dimension equal to $t_i$. Then
$$
        d\le n-\frac{1}{q^k-1}\sum_{i=1}^n(q^{k-t_i}-1).
$$
In particular, if $\C$ is faithful, then 
$$
        d\le \frac{nq^{k-h}(q^h-1)}{q^k-1}.
$$
\end{corollary}

\begin{proof}
Since every nonzero codeword of $\C$ has weight at least $d$, we have
$$
        (q^k-1)d \le \sum_{c\in \C}\wt(c).
$$
The result follows from Theorem~\ref{thm:tot_weight}. If $\C$ is faithful, then $t_i=h$ for every $i$, and therefore
$$
        d\le\frac{n(q^k-q^{k-h})}{q^k-1}=\frac{nq^{k-h}(q^h-1)}{q^k-1}.
$$
\end{proof}

\begin{remark}
Note that the total weight of a nondegenerate additive code depends only on $q$, the length $n$, the dimension $k$, and the dimensions of the subspaces $\pi_i$.
If $t_i=1$ for every $i$, then each $\pi_i$ is a point of
$\PG(k-1,q)$, and we recover the classical formula of the total weight for a nondegenerate linear $[n,k,d]_q$ code:
$$
        \sum_{c\in \C}\wt(c)=n(q^k-q^{k-1}).
$$
If $\C$ is faithful, then $t_i=h$ for every $i$. In this case the formula becomes
$$
        \sum_{c\in \C}\wt(c)=n(q^k-q^{k-h}).
$$
\end{remark}

\subsection{Plotkin bound}
We now derive a Plotkin-type bound for additive codes. Since the classical Plotkin bound holds for nonlinear codes, it 
can also be applied by treating an additive code as a subset of $\F_{q^h}$. In the additive setting, however, one can take into account the vectorial dimensions of the subspaces forming the $h$-projective system. This gives a sharper bound when the code is not faithful.

\begin{theorem}\label{thm:plotkin_ti}
Let $\C$ be a nondegenerate $[n,k/h,d]^h_q$ additive code, and let $\chi(\C)=\{\pi_1,\ldots,\pi_n\}$ be an associated $h$-projective system. Assume that $\pi_i$ has vectorial dimension $t_i$ for each $i\in[n]$. Then, if $d>\sum\limits_{i=1}^n\left(1-\frac{1}{q^{t_i}}\right)$, we have
$$
        q^k
        \le
        \frac{d}
        {d-\sum_{i=1}^n\left(1-\frac{1}{q^{t_i}}\right)}.
$$
\end{theorem}
\begin{proof}
Let
$$
        X:=\sum_{\substack{x,y\in\C\\x\ne y}} d(x,y).
$$
Since the minimum distance of $\C$ is $d$, we have
\begin{equation}\label{eq:lower_plotkin}
    X\ge dq^k(q^k-1).
\end{equation}
On the other hand, since $\C$ is additive, for every $z\in\C$ there are exactly $q^k$ pairs $(x,y)\in\C^2$ such that $x-y=z$. Hence
$$
        \sum_{x,y\in\C}d(x,y)=q^k\sum_{z\in\C}\wt(z).
$$
The terms with $x=y$ contribute zero, so this sum is equal to $X$. By Theorem~\ref{thm:tot_weight}, we obtain
$$
        X=q^k\sum_{z\in\C}\wt(z)=q^k\sum_{i=1}^n(q^k-q^{k-t_i}).
$$
Therefore,
\begin{equation}\label{eq:2Plotkin}
     X = q^{2k}\sum_{i=1}^n\left(1-\frac{1}{q^{t_i}}\right)
\end{equation}
Combining Equations~\eqref{eq:lower_plotkin} and~\eqref{eq:2Plotkin}, we get
$$
        d q^k(q^k-1)\le q^{2k}\sum_{i=1}^n\left(1-\frac{1}{q^{t_i}}\right).
$$
After cancelling $q^k$, this gives
$$
        q^k\left(d-\sum_{i=1}^n\left(1-\frac{1}{q^{t_i}}\right)\right)\le d,
$$
and hence we get the statement.
\end{proof}

We also get the following corollary, which can also be obtained by the classical Plotkin bound applied to $\mC$ as a nonlinear code over $\F_q^h$.

\begin{corollary}\label{cor:plotkin_classic}
Let $\C$ be a nondegenerate $[n,k/h,d]^h_q$ additive code. If $d>n\left(1-\frac{1}{q^h}\right)$, then
$$q^k\leq \frac{q^hd}{q^hd-n(q^h-1)}.$$
\end{corollary}
\begin{proof}
Since $t_i\le h$ for every $i$, we have
$$
        1-\frac{1}{q^{t_i}}\leq 1-\frac{1}{q^h}.
$$
Hence
$$
        \sum_{i=1}^n\left(1-\frac{1}{q^{t_i}}\right)\le n\left(1-\frac{1}{q^h}\right).
$$
Therefore, if $d>n\left(1-\frac{1}{q^h}\right)$, then we have $d>\sum\limits_{i=1}^n\left(1-\frac{1}{q^{t_i}}\right)$. Applying Theorem~\ref{thm:plotkin_ti}, we conclude.
\end{proof}

\begin{remark}
Note that if $\C$ is faithful, then $t_i=h$ for every $i\in[n]$. Therefore
$$
        \sum_{i=1}^n\left(1-\frac{1}{q^{t_i}}\right)
        =
        n\left(1-\frac{1}{q^h}\right).
$$
Hence the bound in Theorem~\ref{thm:plotkin_ti} coincides with the classical Plotkin bound stated in Corollary~\ref{cor:plotkin_classic}.
\end{remark}

Using Theorem~\ref{thm:plotkin_ti}, we can derive an asymptotic bound
which depends on the dimensions of the subspaces $\pi_i$. The proof follows the same idea as in the classical asymptotic Plotkin bound; more precisely, we puncture on a suitable set of coordinates, apply a Plotkin-type bound to the resulting shortened codes, and then count the possible choices on the remaining coordinates. 

We first fix some notation for puncturing and shortening. Let
$\mC\subseteq \F_{q^h}^n$ be an additive code and let $J\subseteq[n]$.
We denote by $\overline J:=[n]\setminus J$ the complement of $J$ in $[n]$.
The \textbf{puncturing of $\mC$ on $J$} is the code
$$
\mC|_J:=\{(c_i)_{i\in J}:c=(c_1,\ldots,c_n)\in\mC\}.
$$
For a vector $a\in\F_{q^h}^{\overline J}$, we define the \textbf{shortening of $\mC$ with respect to $J$ and $a$} as
$$
\mC(J,a):=\{(c_i)_{i\in J}:c\in\mC,\ (c_i)_{i\in\overline J}=a\}.
$$
Thus $\mC(J,a)$ is the slice of $\mC$ obtained by fixing the coordinates
outside $J$ to be equal to $a$, and then puncturing on $J$. Notice that
$\mC(J,a)$ need not be an additive code.

\begin{theorem}\label{thm:asymptotic_plotkin_ti}
Let $(\mC_m)_m$ be a sequence of nondegenerate $[n_m,k_m/h,d_m]^h_q$
additive codes, with $n_m\to+\infty$. For each $m$, let
$\chi(\mC_m)=\{\pi_{1,m},\ldots,\pi_{n_m,m}\}$, where $\pi_{i,m}$ has
vectorial dimension $t_{i,m}$. For every $j\in[h]$, set
$$
n_{j,m}:=|\{i\in[n_m]:t_{i,m}=j\}|.
$$
Assume that, for every $j\in[h]$, the limit
$$
\nu_j:=\lim_{m\to+\infty}\frac{n_{j,m}}{n_m}
$$
exists. Let $R_m:=k_m/n_m$ and $\delta_m:=d_m/n_m$ be the relative rate
and relative distance of $\mC_m$, respectively, and suppose that
$\delta_m\to\delta$. Finally, define
$$
\beta:=\sum_{j=1}^h \nu_j\left(1-\frac{1}{q^j}\right).
$$
Then the following hold.
\begin{enumerate}
    \item If $\delta>\beta$, then $R_m=o(1)$.
    \item If $\delta<\beta$, then
    $$
    \limsup_{m\to+\infty} R_m
    \le
    \sum_{j=1}^h j\nu_j-M_\nu(\delta),
    $$
    where
    $$
    M_\nu(\delta):=
    \sup\left\{
    \sum_{j=1}^h jy_j :
    0\le y_j\le \nu_j,\ 
    \sum_{j=1}^h y_j\left(1-\frac{1}{q^j}\right)\le \delta
    \right\}.
    $$
\end{enumerate}
\end{theorem}

\begin{proof}
For every $m$, set
$$
B_m:=\sum_{i=1}^{n_m}\left(1-\frac{1}{q^{t_{i,m}}}\right).
$$
Grouping the coordinates according to the value of $t_{i,m}$ gives
$$
B_m=\sum_{j=1}^h n_{j,m}\left(1-\frac{1}{q^j}\right),
$$
and hence $B_m/n_m\to\beta$.

Assume first that $\delta>\beta$. Since $\delta_m\to\delta$ and
$B_m/n_m\to\beta$, for $m$ large enough we have $d_m>B_m$. By
Theorem~\ref{thm:plotkin_ti},
$$
q^{k_m}\le \frac{d_m}{d_m-B_m}
=
\frac{\delta_m}{\delta_m-\frac{B_m}{n_m}}.
$$
The right-hand side is bounded independently of $m$, because
$\delta_m-B_m/n_m\to\delta-\beta>0$. Hence $k_m=O(1)$. Since
$n_m\to+\infty$, we obtain $R_m=k_m/n_m=o(1)$.

Assume now that $\delta<\beta$. If $\delta=0$, then $M_\nu(0)=0$, and the
desired bound follows from the trivial inequality
$k_m\le\sum\limits_{i=1}^{n_m}t_{i,m}$, which gives
$\limsup_{m\to+\infty}R_m\le\sum_{j=1}^h j\nu_j$. Thus we may assume that
$\delta>0$.
Let $y_1,\ldots,y_h$ be real numbers such that $0\le y_j\le\nu_j$ for
every $j\in[h]$, and
$$
\sum_{j=1}^h y_j\left(1-\frac{1}{q^j}\right)<\delta.
$$
For $m$ large enough, choose a subset $J_m\subseteq[n_m]$ such that, for
every $j\in[h]$, the set $J_m$ contains $\lfloor y_jn_m\rfloor$
coordinates of vectorial dimension $j$. This is possible since
$n_{j,m}/n_m\to\nu_j$ and $y_j\le\nu_j$.
With this choice, set
$$
B(J_m):=\sum_{i\in J_m}\left(1-\frac{1}{q^{t_{i,m}}}\right).
$$
Then
$$
\frac{B(J_m)}{n_m}
\longrightarrow
\sum_{j=1}^h y_j\left(1-\frac{1}{q^j}\right)
<\delta.
$$
Hence, for $m$ large enough, $d_m>B(J_m)$.

Fix $a\in\F_{q^h}^{\overline{J_m}}$ and consider $\mC_m(J_m,a)$.
If  $\mC_m(J_m,a)\ne\emptyset$, then it has minimum distance at least~$d_m$, because any two codewords in it agree on the coordinates outside $J_m$.
Moreover, for each $i\in J_m$, the set of symbols appearing in the $i$-th coordinate of $\mC_m(J_m,a)$ has size at most $q^{t_{i,m}}$. Thus, the same  argument used in the proof of Theorem~\ref{thm:plotkin_ti} gives
$$
|\mC_m(J_m,a)|
\le
\frac{d_m}{d_m-B(J_m)}.
$$
Since $d_m/n_m\to\delta$ and $B(J_m)/n_m$ has limit strictly smaller than
$\delta$, the right-hand side is bounded independently of $m$. Therefore every nonempty set $\mC_m(J_m,a)$ has size $O(1)$.

It remains to count the possible choices of $a$. For each coordinate
$i$, the set of symbols appearing in the $i$-th coordinate of $\mC_m$ has
size at most $q^{t_{i,m}}$. Hence the number of possible words on
$\overline{J_m}$ is at most
$$
q^{\sum_{i\notin J_m}t_{i,m}}.
$$
Consequently,
$$
|\mC_m|
\le
q^{\sum_{i\notin J_m}t_{i,m}}O(1).
$$
Since $|\mC_m|=q^{k_m}$, taking logarithms base $q$ gives
$$
k_m\le \sum_{i\notin J_m}t_{i,m}+o(n_m).
$$
Dividing by $n_m$, we obtain
$$
R_m
\le
\frac{1}{n_m}\sum_{i=1}^{n_m}t_{i,m}
-
\frac{1}{n_m}\sum_{i\in J_m}t_{i,m}
+o(1).
$$
By the definition of the numbers $\nu_j$, we have
$$
\frac{1}{n_m}\sum_{i=1}^{n_m}t_{i,m}
\longrightarrow
\sum_{j=1}^h j\nu_j.
$$
Moreover, by the choice of $J_m$,
$$
\frac{1}{n_m}\sum_{i\in J_m}t_{i,m}
\longrightarrow
\sum_{j=1}^h jy_j.
$$
Therefore
$$
\limsup_{m\to+\infty} R_m
\le
\sum_{j=1}^h j\nu_j-\sum_{j=1}^h jy_j.
$$
This inequality holds for every $y_1,\ldots,y_h$ satisfying
$0\le y_j\le\nu_j$ and
$$
\sum_{j=1}^h y_j\left(1-\frac{1}{q^j}\right)<\delta.
$$
Taking the supremum over all admissible choices of $y_1,\ldots,y_h$, we obtain the same value if the strict inequality is replaced by
$$
\sum_{j=1}^h y_j\left(1-\frac{1}{q^j}\right)\le\delta.
$$
Indeed, any admissible vector satisfying equality can be approximated by admissible vectors satisfying the strict inequality. Hence
$$
\limsup_{m\to+\infty}R_m
\le
\sum_{j=1}^h j\nu_j-M_\nu(\delta),
$$
as claimed.
\end{proof}

\begin{remark}\label{rem:faithful_plotkin}
If the codes $\mC_m$ are faithful, then $t_{i,m}=h$ for every $i$ and
every $m$. Hence $\nu_h=1$
and $\nu_j=0$ for every $j<h$.
Therefore $\beta=1-\frac{1}{q^h}$. Moreover, for
$0\le \delta<1-\frac{1}{q^h}$, the quantity $M_\nu(\delta)$ becomes
$$
M_\nu(\delta)=\frac{h\delta}{1-\frac{1}{q^h}}.
$$
Thus Theorem~\ref{thm:asymptotic_plotkin_ti} gives
$$
\limsup_{m\to+\infty} R_m
\le
h-\frac{h\delta}{1-\frac{1}{q^h}},
$$
whenever $\delta<1-\frac{1}{q^h}$, while for
$\delta>1-\frac{1}{q^h}$ it gives $R_m=o(1)$. Hence, in the faithful
case, we recover the classical asymptotic Plotkin bound for codes over an
alphabet of size $q^h$, written in terms of the rate $R_m=k_m/n_m$.
\end{remark}

\subsection{One-weight additive codes}
In this subsection, we use Theorem~\ref{thm:tot_weight} and Theorem~\ref{thm:plotkin_ti} in order to derive some results on \emph{one-weight additive codes}. Recall that an additive code is called \textbf{one-weight} or \textbf{constant weight} if all its nonzero codewords have the same Hamming weight. 

\begin{proposition}
Let $\C$ be a nondegenerate $[n,k/h,w]^h_q$ one-weight  additive code, and let $\chi(\C)=\{\pi_1,\ldots,\pi_n\}$ be an $h$-projective system associated with $\mC$. Assume that $\pi_i$ has vectorial dimension~$t_i$. Then
$$
        w=n-\frac{1}{q^k-1}\sum_{i=1}^n(q^{k-t_i}-1).
$$
\end{proposition}

\begin{proof}
Since $\C$ is one-weight, every nonzero codeword has weight $w$. Hence
$$
        \sum_{c\in\C}\wt(c)=(q^k-1)w.
$$
The result follows from Theorem~\ref{thm:tot_weight}.
\end{proof}

In the faithful case, the previous condition becomes especially simple. Indeed, the weight is completely determined by $n,k,h$ and $q$, and the length must satisfy a divisibility condition.

\begin{corollary}
Let $\C$ be an $[n,k/h,w]^h_q$ faithful nondegenerate one-weight additive code. Set $g=\gcd(k,h)$. Then
$$
        \frac{q^k-1}{q^g-1}\mid n.
$$
Moreover, if $n=m\frac{q^k-1}{q^g-1}$, then
$$
        w=m\frac{q^{k-h}(q^h-1)}{q^g-1}.
$$
\end{corollary}
\begin{proof}
Since $\C$ is faithful, Theorem~\ref{thm:tot_weight} gives
$(q^k-1)w=n(q^k-q^{k-h})$.
We can rewrite this as 
$$
        n-w=\frac{n(q^{k-h}-1)}{q^k-1}.
$$
Since $n-w$ is an integer, we have
$$
        q^k-1 \mid n(q^{k-h}-1).
$$
Moreover,
$$
        \gcd(q^k-1,q^{k-h}-1) =q^{\gcd(k,k-h)}-1=        q^{\gcd(k,h)}-1.
$$
Writing $g=\gcd(k,h)$, we obtain
$$
        \frac{q^k-1}{q^g-1}\mid n.
$$
Finally, if $n=m\frac{q^k-1}{q^g-1}$, then,
$$
        w=m\frac{q^{k-h}(q^h-1)}{q^g-1}.
$$
\end{proof}

We can easily translate the one-weight property into the language of $h$-projective systems. This gives a useful geometric characterization.

\begin{proposition}
Let $\C$ be a nondegenerate additive code with associated $h$-projective system $\chi(\C)=\{\pi_1,\ldots,\pi_n\}$. Then $\C$ is one-weight with
nonzero weight $w$ if and only if every hyperplane of $\PG(k-1,q)$
contains exactly $n-w$ elements of $\chi(\C)$, counted with multiplicity.
\end{proposition}

\begin{proof}
Let $G$ be a generator matrix of $\C$. For every nonzero codeword
$c=uG$, with $u\in\F_q^k\setminus\{0\}$, we have $\wt(c)= n-
        |\{i\in[n]:\pi_i\subseteq \langle u\rangle_{\Fq}^\perp\}|$.
The hyperplanes of $\PG(k-1,q)$ are precisely the spaces
$\langle u_\rangle{\Fq}^\perp$, with $u\in\F_q^k\setminus\{0\}$. Hence all
nonzero codewords have weight $w$ if and only if every hyperplane
contains exactly $n-w$ elements of $\chi(\C)$, counted with multiplicity.
\end{proof}

One-weight additive codes also attain the refined Plotkin bound from Theorem~\ref{thm:plotkin_ti}.

\begin{proposition}
Let $\C$ be an $[n,k/h,w]^h_q$ nondegenerate one-weight  additive code, and let $\chi(\C)=\{\pi_1,\ldots,\pi_n\}$ be an $h$-projective system associated with $\mC$. Assume that $\pi_i$ has vectorial dimension $t_i$ for every $i\in[n]$. Then
$$w=\frac{q^k}{q^k-1}\sum_{i=1}^n\left(1-\frac{1}{q^{t_i}}\right).
$$
In particular, $\C$ meets the bound of Theorem~\ref{thm:plotkin_ti} with equality.
\end{proposition}

\begin{proof}
Since $\C$ is one-weight, every nonzero codeword has weight $w$. Hence
$$\sum_{c\in\C}\wt(c)=(q^k-1)w.$$
On the other hand, by Theorem~\ref{thm:tot_weight},
$$\sum_{c\in\C}\wt(c)=\sum_{i=1}^n(q^k-q^{k-t_i})=q^k\sum_{i=1}^n\left(1-\frac{1}{q^{t_i}}\right).
$$
Therefore,
$$
        (q^k-1)w=q^k\sum_{i=1}^n\left(1-\frac{1}{q^{t_i}}\right),
$$
and hence
$$
        w=\frac{q^k}{q^k-1}\sum_{i=1}^n\left(1-\frac{1}{q^{t_i}}\right).
$$
In particular, $w>\sum\limits_{i=1}^n\left(1-\frac{1}{q^{t_i}}\right)$ and
$$
        \frac{w}{w-\sum\limits_{i=1}^n\left(1-\frac{1}{q^{t_i}}\right)}=q^k.
$$
Since the minimum distance of $\C$ is $d=w$, this shows that the bound in Theorem~\ref{thm:plotkin_ti} is attained with equality.
\end{proof}

The previous characterization also gives explicit examples. We now show that there exist faithful one-weight additive codes which are not equivalent to any $\F_{q^h}$-linear code. These codes could be considered as an additive version of \emph{simplex codes}. 

\begin{theorem}\label{thm:simplex_add}
Let $2\le h<k$. Then there exists an $[n,k/h,w]^h_q$ faithful one-weight additive code which is not equivalent to any
$\F_{q^h}$-linear code with
$$
        n=\binom{k}{h}_q, \qquad 
        w=q^{k-h}\binom{k-1}{h-1}_q.
$$
\end{theorem}

\begin{proof}
Let $\chi$ be the set of all $(h-1)$-spaces of $\PG(k-1,q)$. Then $|\chi|=\binom{k}{h}_q$.
Moreover, every hyperplane of $\PG(k-1,q)$ contains exactly $\binom{k-1}{h}_q$ elements of $\chi$. Hence, the additive code associated with $\chi$ is one-weight, with nonzero weight
$$
        w=\binom{k}{h}_q-\binom{k-1}{h}_q=q^{k-h}\binom{k-1}{h-1}_q.
$$
Since every element of $\chi$ has vectorial dimension $h$, the code is faithful.
It remains to show that this code is not equivalent to an $\F_{q^h}$-linear code. If $h\nmid k$, this is immediate, since an $\F_{q^h}$-linear code has $\F_q$-dimension divisible by $h$.
Assume now that $h\mid k$. If the code were equivalent to an
$\F_{q^h}$-linear code, then its associated projective system would be contained in a Desarguesian $(h-1)$-spread of $\PG(k-1,q)$. However, $\chi$ is the set of all $(h-1)$-spaces of $\PG(k-1,q)$, and since $h<k$, it is not contained in a single Desarguesian spread. Therefore, the associated additive code is not equivalent to an $\F_{q^h}$-linear code.
\end{proof}

\section{Minimal additive codes and additive strong blocking sets}\label{sec:min_codes}

In this section, we introduce the notion of \emph{minimal additive codes} and we describe  their geometric counterpart, by generalizing the concept of outer strong blocking sets and outer minimal codewords. Outer minimal codes are $\F_{q^h}$-linear subspaces of $\F_{q^h}^n$ whose concatenation with minimal codes is minimal. In \cite{alfarano2024outer}, it is shown that they are in one-to-one correspondence with \emph{outer strong blocking sets}, which are sets of points in $\PG(k/h-1,q^h)$ whose field reduction is a strong blocking set.

We start by recalling the classical notions for linear codes that motivate our definitions.

\begin{definition}
An $\F_q$-linear code $\C\subseteq \F_q^n$ is \textbf{minimal} if, for every
$c,c'\in\C\setminus\{0\}$,
$$\supp(c')\subseteq\supp(c)\quad\Leftrightarrow \quad c'=\lambda c \text{ for some }\lambda\in\F_q^*.$$
\end{definition}

\begin{definition}
A \textbf{strong blocking set} in $\PG(k-1,q)$ is a multiset of points $\M$, not
contained in a hyperplane, such that for all hyperplanes $H,H'$ of
$\PG(k-1,q)$,
$$\M\cap H\subseteq \M\cap H'\quad\Leftrightarrow\quad H=H'.
$$
\end{definition}
In \cite{alfarano2022geometric, tang2021full} it is shown that (monomial) equivalence classes of nondegenerate minimal linear codes are in one-to-one correspondence with projective equivalence classes of strong blocking sets. This correspondence has been exploited in the literature for providing bounds and constructions.

We now move to the additive setting.

\begin{definition}\label{def:min_codewords}
Let $\C$ be an additive code. We say that a codeword $c\in\C\setminus\{0\}$ is \textbf{minimal} if, for every
$c'\in\C\setminus\{0\}$ satisfying
$\supp(c')\subseteq\supp(c)$
and
$\frac{c_i'}{c_i}\in\F_q^\ast$ for every $i\in\supp(c')$,
there exists $\lambda\in\F_q^*$ such that $c'=\lambda c$.
\end{definition}

\begin{remark}
Definition~\ref{def:min_codewords} coincides with
\cite[Definition~23]{alfarano2024outer}, which introduce \emph{minimal codewords} of an $\F_{q^h}$-linear code in $\F_{q^h}^n$. In particular, minimal codewords
in additive codes are outer minimal in the sense of~\cite{alfarano2024outer}.
We simply use the word ``minimal'' in order not to distinguish the $\F_{q^h}$-linear case from the general additive case. We will discuss in detail this relation in Section~\ref{sec:Linear_case}.
\end{remark}

\begin{definition}
    We say that an additive code $\mathcal{C}$ is \textbf{minimal} if all its nonzero codewords are minimal. 
\end{definition}

\begin{example}
    Consider $\F_4=\F_2(\omega)$, with $\omega$ a root of $x^2+x+1\in\F_2[x]$.  Let $\C$ be the $[4,2/2,3]_2^2$ additive code defined by the generator matrix 
    $$G=\begin{pmatrix}
1 & 1 & \omega & \omega\\
1 & \omega & 0 & 1
\end{pmatrix}.$$
We have that $\C=\{(0,0,0,0),(1,1,\omega,\omega),(1,\omega,0,1),(0,\omega^2,\omega,\omega^2)\}$. The code $\C$ is minimal.
\end{example}

We provide a straightforward characterization of minimal codewords. 

\begin{proposition}\label{prop:charact_min_codewords}\label{prop:minimal_codeword_characterization} Let $\C$ be an $[n,k/h]_q^h$ additive code. A codeword $c \in \C\backslash\{0\}$ is minimal if and only if $\forall c'\in \mathcal{C}\backslash\{0\}$ such that $\supp(c')\subsetneq \supp(c)$, $\exists \phantom{i} i \in \supp(c')$ such that $\frac{c_i'}{c_i}\notin \F_q$.
\end{proposition}
\begin{proof}
Suppose first that $c$ is minimal. Let $c'\in\C\backslash\{0\}$ be such that $\supp(c')\subsetneq \supp(c)$. Assume by contradiction that, for every $i\in\supp(c')$, we have $\frac{c_i'}{c_i}\in \F_q$. Since $i\in\supp(c')$, we have $c_i'\neq 0$, and hence $\frac{c_i'}{c_i}\in \F_q^\ast$ for every $i\in\supp(c')$. By minimality of $c$, there exists $\lambda\in\F_q^\ast$ such that $c'=\lambda c$. This contradicts the strict inclusion $\supp(c')\subsetneq \supp(c)$. Therefore, there exists $i\in\supp(c')$ such that $\frac{c_i'}{c_i}\notin \F_q$.

Conversely, suppose that $c\in\C\backslash\{0\}$ satisfies the stated condition. Let $c'\in\C\backslash\{0\}$ be such that $\supp(c')\subseteq \supp(c)$ and, for every $i\in\supp(c')$, $\frac{c_i'}{c_i}\in \F_q^\ast$. By the assumption, the inclusion $\supp(c')\subseteq\supp(c)$ cannot be strict. Hence, $\supp(c')=\supp(c)$. Pick $i\in\supp(c)$ and define $\lambda:=\frac{c_i'}{c_i}\in\F_q^\ast$ and $u:=c'-\lambda c\in\C$.
If $u\ne 0$, then $\supp(u)\subsetneq \supp(c)$, since the $i$-th coordinate of $u$ is zero. Therefore, by the assumption, there would exist $j\in\supp(u)$ such that $\frac{u_j}{c_j}\notin\F_q$. But
$$
\frac{u_j}{c_j}
=
\frac{c_j'-\lambda c_j}{c_j}
=
\frac{c_j'}{c_j}-\lambda
\in\F_q,
$$
which is a contradiction. Hence $u=0$, so $c'=\lambda c$, with $\lambda\in\F_q^\ast$. Therefore, the codeword $c$ is minimal.
\end{proof}

As a consequence, we obtain that one-weight additive codes are minimal. 

\begin{proposition}\label{prop:one_weight_additive_minimal}
Let $\C$ be a one-weight additive code. Then $\C$ is a minimal additive code.
\end{proposition}

\begin{proof}
Let $c\in\C\setminus\{0\}$. We prove that $c$ is minimal. By
Proposition~\ref{prop:minimal_codeword_characterization}, it is enough to show that
there is no nonzero codeword $c'\in\C$ such that
$\supp(c')\subsetneq\supp(c)$. Indeed, if such a codeword existed, then
$\wt(c')<\wt(c)$, contradicting the fact that $\C$ is one-weight. Therefore every
nonzero codeword of $\C$ is minimal, and hence $\C$ is a minimal additive code.
\end{proof}

The following result is essential for describing a geometric framework for the study of minimal additive codes.

\begin{proposition}\label{prop:minimality_inv}
Let $\mC, \mC'$ be two $[n,k/h,d]_q^h$ equivalent additive codes. Then $\mC$ is minimal if and only if $\mC'$ is minimal.
\end{proposition}
\begin{proof}
It is enough to prove one implication. Suppose that $\mC'$ is minimal.
Since $\mC$ and $\mC'$ are equivalent, there exist a permutation $\sigma \in \mathcal{S}_n$ and $\Fq$-linear automorphisms $\varphi_1,\ldots,\varphi_n$ of $\Fqh$ such that
$$\mC'=\{(\varphi_1(c_{\sigma(1)}),\ldots,\varphi_n(c_{\sigma(n)})) \; : \; c=(c_1,\ldots,c_n)\in \mC\}.$$
Let $c,c' \in \mC \backslash \{0\}$ be such that $\supp(c') \subseteq \supp(c)$ and, for every $i \in \supp(c')$, $\frac{c_i'}{c_i} \in \Fq^\ast$. We want to prove that there exists $\lambda \in \Fq^\ast$ such that $c'= \lambda c$.
Let
$x:=(\varphi_1(c_{\sigma(1)}),\ldots,\varphi_n(c_{\sigma(n)}))$
and $x':=(\varphi_1(c'_{\sigma(1)}),\ldots,\varphi_n(c'_{\sigma(n)}))$
be the corresponding nonzero codewords of $\mC'$. It is clear that $\supp(c')\subseteq \supp(c)$ implies $\supp(x')\subseteq \supp(x)$.
Let $i \in \supp(x')$. Then $\sigma(i)\in\supp(c')$, and there exists $\mu_i\in\Fq^\ast$ such that
$c'_{\sigma(i)}=\mu_i c_{\sigma(i)}$. Since $\varphi_i$ is $\Fq$-linear, we have
$$
x_i'=\varphi_i(c'_{\sigma(i)})
=\varphi_i(\mu_i c_{\sigma(i)})
=\mu_i\varphi_i(c_{\sigma(i)})
=\mu_i x_i.
$$
Therefore, $\frac{x_i'}{x_i}\in\Fq^\ast$ for every $i\in\supp(x')$. Hence, by the minimality of $\mC'$, there exists $\lambda \in \Fq^\ast$ such that $x'=\lambda x$. Since each $\varphi_i$ is injective, this implies that $c'=\lambda c$. Therefore, $\mC$ is minimal.
\end{proof}

\subsection{The geometry of minimal additive codes}
We now introduce the geometric object corresponding to minimal additive codes. It generalizes outer strong blocking sets; see
\cite[Definition~18]{alfarano2024outer}.

\begin{definition}
An $h$-$[n,k,d]_q$ projective system $\{\pi_1,\ldots,\pi_n\}$ is an \textbf{additive strong blocking set} in $\PG(k-1,q)$ if $\bigcup\limits_{i=1}^n \pi_i$ is a strong blocking set in $\PG(k-1,q)$.
\end{definition}

\begin{remark}
Additive strong blocking sets and outer strong blocking sets are related but distinct. Outer strong blocking sets arise from points in $\PG(k/h-1,q^h)$ whose field reduction gives rise to subspaces of a Desarguesian $(h-1)$-spread in $\PG(k-1,q)$, whereas additive strong blocking sets are defined more generally, requiring only that the union of their subspaces forms a strong blocking set in $\PG(k-1,q)$. In particular, every outer strong blocking set gives rise to an additive strong blocking set, but not every additive strong blocking set gives rise to an outer strong blocking set. We will clarify the distinction in Section~\ref{sec:Linear_case}.
\end{remark}

\begin{proposition}\label{prop:sbs_invariant}
Let $\chi$ and $\psi$ be two equivalent $h$-$[n,k,d]_q$ projective systems. Then $\chi$ is an additive strong blocking set if and only if $\psi$ is an additive strong blocking set.
\end{proposition}

\begin{proof}
Let us denote $\chi=\{\pi_1,\ldots,\pi_n\}$ and $\psi=\{\mu_1,\ldots,\mu_n\}$. Suppose that $\chi$ is an additive strong blocking set. Since $\chi$ and $\psi$ are equivalent, there exists a projective transformation $\varphi$ mapping $\chi$ to $\psi$. Let $H,H'$ be two hyperplanes such that
$$
\left(\bigcup_i \mu_i\right)\cap H \subseteq \left(\bigcup_i \mu_i\right)\cap H'.
$$
Equivalently,
$$
\left(\bigcup_i \varphi(\pi_i)\right)\cap H \subseteq \left(\bigcup_i \varphi(\pi_i)\right)\cap H'.
$$
Applying $\varphi^{-1}$ to this inclusion, we obtain
$$
\left(\bigcup_i \pi_i\right)\cap \varphi^{-1}(H) \subseteq \left(\bigcup_i \pi_i\right)\cap \varphi^{-1}(H').
$$
Since $\varphi^{-1}(H)$ and $\varphi^{-1}(H')$ are hyperplanes and $\chi$ is an additive strong blocking set, we get $\varphi^{-1}(H)=\varphi^{-1}(H')$, and  hence $H=H'$. Therefore, $\psi$ is an additive strong blocking set.
The converse follows by applying the same argument to the inverse projective transformation.
\end{proof}

In the rest of this section, we show that equivalence classes of minimal additive codes are in one-to-one correspondence with equivalence classes of additive strong blocking sets. First of all, observe that, by Propositions~\ref{prop:minimality_inv} and~\ref{prop:sbs_invariant}, the properties of being minimal and of being an additive strong blocking set are invariant under equivalence. We introduce the following useful notation.

\begin{notation}
Let $\C$ be an $[n,k/h,d]_q^h$ nondegenerate additive code. We define an $\Fq$-linear code $\widehat{\C} \subseteq \Fq^{nq^h}$ as follows.
Let $G$ be a generator matrix of $\C$, and let $\widetilde{G}$ be the matrix defined in Eq.~\eqref{eq:Gtilde}. With a slight abuse of notation, we denote the columns of $\widetilde{G}$ by
$$
(P_1^1,\ldots,P_h^1 \mid \ldots \mid P_1^n,\ldots,P_h^n),
$$
where the columns are grouped into $n$ blocks of size $h$.  We denote by $\widetilde{\C}$ the $\Fq$-linear code generated by the rows of $\widetilde{G}$.
For every $i\in[n]$, let $U_i$ be the $\Fq$-subspace of $\Fq^k$ generated by $P_1^i,\ldots,P_h^i$. We extend the matrix $\widetilde{G}$ to a matrix
$$\widehat{G}=\left(
\widehat{G}_1 \mid \ldots \mid \widehat{G}_n\right),
$$
where each block $\widehat{G}_i$ has $q^h$ columns. More precisely, the columns of $\widehat{G}_i$ are the vectors of $U_i$, each appearing at least once, and repeated zero columns are added if necessary so that the block has exactly $q^h$ columns. We choose the ordering in such a way that the first $h$ columns of $\widehat G_i$
are $P_1^i,\ldots,P_h^i$. Finally, we define $\widehat{\C}$ as the $\Fq$-linear code generated by the rows of $\widehat{G}$.
This construction is well-defined up to equivalence of linear codes, since different choices of the ordering of the columns in each block, or of the repeated zero columns, only give equivalent codes.
\end{notation}

\begin{lemma}\label{lem:minimal_extended}
Let $\C$ be a $[n,k/h,d]_q^h$ minimal additive code. Then $\widehat{\C}\subseteq \F_q^{nq^h}$ is a $k$-dimensional minimal $\F_q$-linear code. 
\end{lemma}

\begin{proof}
The code $\widehat{\C}$ has dimension $k$, since $\widehat{G}$ and $\widetilde{G}$ have the same row space over $\Fq$, and $\widetilde{G}$ has rank $k$. Let
$$
x:=(x_1^1,\ldots,x_{q^h}^1,\ldots,x_1^n,\ldots,x_{q^h}^n),
\quad
x':=({x_1^1}',\ldots,{x_{q^h}^1}',\ldots,{x_1^n}',\ldots,{x_{q^h}^n}')
$$
be two nonzero codewords of $\widehat{\C}$ such that $\supp(x')\subseteq\supp(x)$.
By the construction of $\widehat{G}$, the first $h$ coordinates in each block give codewords
$$
\tilde{c}:=(x_1^1,\ldots,x_h^1,\ldots,x_1^n,\ldots,x_h^n)\in\widetilde{\C}\backslash\{0\}
$$
and
$$
\tilde{c}':=({x_1^1}',\ldots,{x_h^1}',\ldots,{x_1^n}',\ldots,{x_h^n}')\in\widetilde{\C}\backslash\{0\}.
$$
Thus, with respect to the basis $\mathcal{B}=\{\alpha_1,\ldots,\alpha_h\}$ of $\F_{q^h}$ over $\Fq$, they correspond to the codewords
$$
c:=\left(\sum_{i=1}^h x_i^1\alpha_i,\ldots,\sum_{i=1}^h x_i^n\alpha_i\right)\in\C\backslash\{0\}
$$
and
$$
c':=\left(\sum_{i=1}^h {x_i^1}'\alpha_i,\ldots,\sum_{i=1}^h {x_i^n}'\alpha_i\right)\in\C\backslash\{0\}.
$$

We prove that there exists $\lambda \in \Fq^\ast$ such that $x'= \lambda x$. First, since $\supp(x') \subseteq \supp(x)$, we have $\supp(\tilde{c}') \subseteq \supp(\tilde{c})$. Therefore, $\supp(c') \subseteq \supp(c)$. Indeed, if $c_j=\sum\limits_{i=1}^h x_i^j \alpha_i=0$, then $x_i^j=0$ for every $i \in [h]$, since $\mathcal{B}$ is a basis of $\F_{q^h}$ over $\F_q$. Thus, ${x_i^j}'=0$ for every $i \in [h]$, precisely because $\supp(\tilde{c}') \subseteq \supp(\tilde{c})$, and then $c_j'=\sum\limits_{i=1}^h {x_i^j}' \alpha_i=0$.
Moreover, let $j \in \supp(c')$. We show that $\frac{c_j'}{c_j}\in \Fq^\ast$, that is, there exists $\lambda_j \in\Fq^\ast$ such that ${x_i^j}'=\lambda_j x_i^j$ for every $i \in [h]$. Suppose by contradiction that this is not the case. Then, since the two $h$-tuples $(x_1^j,\ldots,x_h^j)$ and $({x_1^j}',\ldots,{x_h^j}')$ are not proportional over $\Fq$, there exists an $\Fq$-linear combination of the first $h$ coordinates of the $j$-th block which is zero on $x$ and nonzero on $x'$. By construction, this linear combination appears as a coordinate in the $j$-th block of $\widehat{G}$. Hence, in the corresponding coordinate, the codeword $x$ is zero while the codeword $x'$ is nonzero, contradicting $\supp(x')\subseteq \supp(x)$. Therefore, there exists $\lambda_j \in\Fq^\ast$ such that ${x_i^j}'=\lambda_j x_i^j$ for every $i \in [h]$. Hence, $c_j'=\lambda_j c_j$, and then $\frac{c_j'}{c_j}\in\Fq^\ast$. Since $\C$ is a minimal additive code, there exists $\lambda \in \Fq^\ast$ such that $c'=\lambda c$. This means that, for every block, the first $h$ coordinates of $x'$ are equal to $\lambda$ times the first $h$ coordinates of $x$. Since the remaining coordinates in each block are $\Fq$-linear combinations of the first $h$ coordinates, we also have $x'=\lambda x$. Therefore, the linear code $\widehat{\C}$ is minimal. Notice that $\widehat{\C}$ is in general degenerate, since each block of $\widehat{G}$ contains the zero vector of $U_j$. This does not affect the argument above. When we use $\widehat{\C}$ to obtain a projective system, we simply delete the zero columns of $\widehat{G}$. The resulting punctured code is still minimal, and its associated projective system is precisely, as a set, $\bigcup\limits_{j=1}^n \pi_j.$
\end{proof}

The following is the main result of this section.

\begin{theorem}\label{thm:minimal_additive_sbs}  
Equivalence classes of $[n,k/h,d]_q^h$ nondegenerate minimal additive codes are in one-to-one correspondence with equivalence classes of $h$-$[n,k,d]_q$ projective systems that are additive strong blocking sets.
\end{theorem} 

\begin{proof}
The statement of the theorem is well-defined because of Propositions~\ref{prop:minimality_inv} and~\ref{prop:sbs_invariant}. Let $\C$ be an $[n,k/h,d]_q^h$ additive code and let $\chi=\{\pi_1,\ldots,\pi_n\}$ be its associated $h$-projective system. Let
$$
G=\left(
\begin{array}{c|c|c}
g_1 & \cdots & g_n
\end{array}
\right)
$$
be a generator matrix of $\C$.
If $\C$ is minimal, then, by Lemma~\ref{lem:minimal_extended}, $\widehat{\C}$ is a minimal linear code. Deleting the zero coordinates of $\widehat{\C}$, we obtain a nondegenerate
minimal linear code. Hence, by the classical correspondence between minimal linear codes and strong blocking sets, its associated projective system $\mathcal{P}(\widehat{\C})$ is a
strong blocking set. By the construction of $\widehat{G}$, this projective system is, as a set,
$$\mathcal{P}(\widehat{\C})=\bigcup_{i=1}^n \pi_i,
$$
we have that $\chi$ is an additive strong blocking set.

Conversely, suppose that $\chi$ is an additive strong blocking set. Then, for all hyperplanes $H,H'$ of $\PG(k-1,q)$, we have
\begin{equation}\label{eq:sbs-hyperplanes}
\left(\bigcup_{i=1}^n \pi_i\right)\cap H
\subseteq
\left(\bigcup_{i=1}^n \pi_i\right)\cap H'
\quad \Longleftrightarrow \quad
H=H'.
\end{equation}
Let $c=uG$ and $c'=u'G$ be two nonzero codewords of $\C$ such that $\supp(c')\subseteq\supp(c)$ and, for every $i\in\supp(c')$, $\frac{c_i'}{c_i}\in\Fq^\ast$. We have to show that there exists $\lambda\in\Fq^\ast$ such that $c'=\lambda c$.

Let $H=\langle u\rangle_{\Fq}^\perp$ and $H'=\langle u'\rangle_{\Fq}^\perp$.
It is enough to prove that $H=H'$. Indeed, if $H=H'$, then $\langle u\rangle_{\Fq}=\langle u'\rangle_{\Fq}$, and so $u'=\lambda u$ for some $\lambda\in\Fq^\ast$. Hence $c'=u'G=\lambda uG=\lambda c$.
Since $\supp(c')\subseteq\supp(c)$, we have
$$
\{i\in[n] : \pi_i \subseteq H\}
\subseteq
\{i\in[n] : \pi_i \subseteq H'\}.
$$
Indeed, $\pi_i\subseteq H$ is equivalent to $c_i=0$, and $\pi_i\subseteq H'$ is equivalent to $c_i'=0$.

We now show that
$$\left(\bigcup_{i=1}^n \pi_i\right)\cap H \subseteq
\left(\bigcup_{i=1}^n\pi_i\right)\cap H'.$$
Let $P\in \left(\bigcup_{i=1}^n \pi_i\right)\cap H$. Then $P\in\pi_i$ for some $i$, and $P\in H$. If $\pi_i\subseteq H$, then, by the above inclusion, $\pi_i\subseteq H'$, and therefore $P\in H'$.

Suppose now that $\pi_i\not\subseteq H$. Then $c_i\neq 0$. Since $P\in H\cap\pi_i$, we need to prove that $P\in H'$. If $c_i'=0$, then $\pi_i\subseteq H'$, and there is nothing to prove. If $c_i'\neq 0$, then $i\in\supp(c')$, and by assumption there exists $\lambda_i\in\Fq^\ast$ such that $c_i'=\lambda_i c_i$. This means that the linear form associated with $u'-\lambda_i u$ vanishes on the whole subspace $\pi_i$, namely
$\pi_i\subseteq \langle u'-\lambda_i u\rangle_{\Fq}^\perp$.
Since $P\in H=\langle u\rangle_{\Fq}^\perp$, we get $P\in \langle u'\rangle_{\Fq}^\perp=H'$.
Therefore the desired inclusion holds.
By Eq.~\eqref{eq:sbs-hyperplanes}, we obtain $H=H'$. Hence $c'=\lambda c$ for some $\lambda\in\Fq^\ast$, and therefore $\C$ is a minimal additive code.
\end{proof}

\begin{example}
Consider $\F_4=\F_2(\omega)$, where $\omega$ is a root of $x^2+x+1\in\F_2[x]$.
Let $\mC$ be the $[3,3/2,2]_2^2$ additive code defined by the generator matrix
$$
G=
\begin{pmatrix}
1 & 1 & 1 \\
\omega & \omega^2 & \omega \\
0 & \omega & \omega
\end{pmatrix}.
$$
Choose $\{1,\omega\}$ as a basis of $\F_4$ over $\F_2$. The expanded generator matrix of $\mC$ is
$$
\widetilde{G}=
\left(
\begin{array}{cc|cc|cc}
1 & 0 & 1 & 0 & 1 & 0  \\
0 & 1 & 1 & 1 & 0 & 1  \\
0 & 0 & 0 & 1 & 0 & 1
\end{array}
\right).
$$
Thus, a $2$-$[3,3,2]_2$ projective system associated with $\mC$ is $\chi=\{\pi_1,\pi_2,\pi_3\}$, where
\begin{itemize}[itemsep=0.001em]
    \item $\pi_1=\{[1:0:0],[0:1:0],[1:1:0]\}$,
    \item $\pi_2=\{[1:1:0],[0:1:1],[1:0:1]\}$,
    \item $\pi_3=\{[1:0:0],[0:1:1],[1:1:1]\}$.
\end{itemize}
Hence, $\bigcup_{i=1}^3 \pi_i$ is the union of three distinct lines of $\PG(2,2)$ which form a strong blocking set; see, for instance,~\cite[Theorem 5.3]{alfarano2022geometric}. Therefore, $\chi$ is an additive strong blocking set in $\PG(2,2)$, and so the code $\mC$ is minimal. Note also that this example is not obtained by field reduction from an $\F_4$-linear code:
the three lines do not form a subset of a Desarguesian line spread of
$\PG(2,2)$, since such a spread does not exist in $\PG(2,2)$.
\end{example}

\subsection{The linear case and outer strong blocking sets}\label{sec:Linear_case}

In this section we compare additive strong blocking sets with the outer strong blocking
sets introduced in~\cite{alfarano2024outer}. This clarifies in which sense the additive
setting is more general.

Assume that $h\mid k$, and write $k=Kh$. An $\F_{q^h}$-linear $[n,K,d]_{q^h}$ code
can also be regarded as an additive $[n,k/h,d]^h_q$ code.
Geometrically, the usual projective system in $\PG(K-1,q^h)$ is sent by \emph{field
reduction} to a collection of $(h-1)$-spaces in $\PG(k-1,q)$ belonging to a
Desarguesian spread. Thus the outer framework is contained in the additive framework
as the special case in which all coordinate subspaces arise from field reduction.

We briefly recall the field reduction map in the notation used in this paper.
Let $\omega$ be a primitive element of $\F_{q^h}$ and fix the $\F_q$-basis $\mathcal B=\{1,\omega,\ldots,\omega^{h-1}\}$
of $\F_{q^h}$ over $\F_q$. Let
$$
\Phi:\F_{q^h}\longrightarrow \F_q^h
$$
be the $\F_q$-linear isomorphism defined by
$$
\Phi\left(\sum_{i=0}^{h-1}x_i\omega^i\right)=(x_0,\ldots,x_{h-1}).
$$
We extend it coordinatewise to an $\F_q$-linear isomorphism
$$
\bar\Phi:\F_{q^h}^K\longrightarrow \F_q^{Kh},
\qquad
(v_1,\ldots,v_K)\longmapsto(\Phi(v_1),\ldots,\Phi(v_K)).
$$

Let
$f=X^h+\sum\limits_{i=0}^{h-1}a_iX^i$
be the minimal polynomial of $\omega$. Since we use row vectors, we take
the companion matrix associated with multiplication by $\omega$ to be
$$
A=
\begin{pmatrix}
0 & 1 & 0 & \cdots & 0\\
0 & 0 & 1 & \cdots & 0\\
\vdots & \vdots & \vdots & \ddots & \vdots\\
0 & 0 & 0 & \cdots & 1\\
-a_0 & -a_1 & -a_2 & \cdots & -a_{h-1}
\end{pmatrix}
\in \M_h(\F_q).
$$
Thus $\Phi(x\omega)=\Phi(x)A$
for every $x\in\F_{q^h}$. For $y\in\F_{q^h}$, define
$$
A(y):=
\begin{cases}
A^r, & \text{if } 0\neq y=\omega^r,\\
0, & \text{if } y=0.
\end{cases}
$$
Then, for every $x,y\in\F_{q^h}$, we have
\begin{equation}\label{eq:propertyphi}
\Phi(xy)=\Phi(x)A(y).
\end{equation}
Let $P=\langle u\rangle_{\F_{q^h}}\in\PG(K-1,q^h)$, with
$u\in\F_{q^h}^K\setminus\{0\}$. The field reduction of $P$ is defined as
$$
\mathcal F(P):=
\PG\big(\bar\Phi(\langle u\rangle_{\F_{q^h}})\big)
=
\PG\big(\{\bar\Phi(\lambda u):\lambda\in\F_{q^h}\}\big)
\subseteq \PG(Kh-1,q).
$$
Since $\langle u\rangle_{\F_{q^h}}$ is $1$-dimensional over $\F_{q^h}$, the space
$\bar\Phi(\langle u\rangle_{\F_{q^h}})$ is $h$-dimensional over $\F_q$. Hence
$\mathcal F(P)$ is an $(h-1)$-dimensional projective subspace of $\PG(Kh-1,q)$.
The set $\{\mathcal{F}(P):P\in\PG(K-1,q^h)\}$ forms a Desarguesian
$(h-1)$-spread of $\PG(Kh-1,q)$; see, for instance,
\cite{Lavrauw2028Field,bartoli2023small,alfarano2024outer}.

\begin{proposition}\label{prop:linear_case_field_reduction}
Let $h\mid k$, and write $k=Kh$. Let $\mathcal D$ be a nondegenerate
$\F_{q^h}$-linear $[n,K,d]_{q^h}$ code, and let
$\mathcal P(\mathcal D)=\{P_1,\ldots,P_n\}$ be its associated projective system in
$\PG(K-1,q^h)$. Regard $\mathcal D$ as an additive $[n,k/h,d]^h_q$ code $\mC$.
Then $\{\mathcal F(P_1),\ldots,\mathcal F(P_n)\}\subseteq \PG(k-1,q)$ is an $h$-projective system associated with $\mC$.
In particular, $\mC$ is faithful, and its associated subspaces belong to a Desarguesian $(h-1)$-spread.
\end{proposition}

\begin{proof}
Let $M=
\left(
\begin{array}{c|c|c|c}
\textcolor{red}{c_1} &
\textcolor{ForestGreen}{c_2} &
\cdots &
\textcolor{blue}{c_n}
\end{array}\right)=(m_{i,j}) \in \M_{K,n}(\F_{q^h})$
be a generator matrix of $\mathcal D$ over $\F_{q^h}$. Thus, considering
$\mC$ and $\mathcal{D}$ equal as sets, we have that
$\mC=\mathcal D=\{yM : y \in \F_{q^h}^K\}$. Moreover, $P_i$ is the point of
$\PG(K-1,q^h)$ represented by $\langle c_i\rangle_{\F_{q^h}}$.
For each $i\in[n]$, the field reduction of $P_i$ is
$$
\mathcal F(P_i)
=
\PG\big(\bar\Phi(\langle c_i\rangle_{\F_{q^h}})\big)
=
\PG\big(\{\bar\Phi(\lambda c_i):\lambda\in\F_{q^h}\}\big)
\subseteq \PG(k-1,q).
$$
This is an $(h-1)$-space, since $\langle c_i\rangle_{\F_{q^h}}$ is
$1$-dimensional over $\F_{q^h}$ and hence $h$-dimensional over $\F_q$
after applying $\bar\Phi$.
Now, expand the entries of $M$ with respect to the basis
$\mathcal B=\{1,\omega,\ldots,\omega^{h-1}\}$ of $\F_{q^h}$ over $\F_q$.
The expanded matrix has the form
$$
\widetilde G=
\left(
\begin{array}{c|c|c|c}
\textcolor{red}{\widetilde G_1} &
\textcolor{ForestGreen}{\widetilde G_2} &
\cdots &
\textcolor{blue}{\widetilde G_n}
\end{array}
\right)=\begin{pmatrix}
\textcolor{red}{b^{1}_{1,1}} & \cdots & \textcolor{red}{b^{h}_{1,1}} &
\textcolor{ForestGreen}{b^{1}_{1,2}} & \cdots & \textcolor{ForestGreen}{b^{h}_{1,2}} & \cdots &
\textcolor{blue}{b^{1}_{1,n}} & \cdots & \textcolor{blue}{b^{h}_{1,n}} \\
\textcolor{red}{b^{1}_{2,1}} & \cdots & \textcolor{red}{b^{h}_{2,1}} &
\textcolor{ForestGreen}{b^{1}_{2,2}} & \cdots & \textcolor{ForestGreen}{b^{h}_{2,2}} & \cdots &
\textcolor{blue}{b^{1}_{2,n}} & \cdots & \textcolor{blue}{b^{h}_{2,n}} \\
\vdots & & \vdots & \vdots & & \vdots & & \vdots & & \vdots \\
\textcolor{red}{b^{1}_{k,1}} & \cdots & \textcolor{red}{b^{h}_{k,1}} &
\textcolor{ForestGreen}{b^{1}_{k,2}} & \cdots & \textcolor{ForestGreen}{b^{h}_{k,2}} & \cdots &
\textcolor{blue}{b^{1}_{k,n}} & \cdots & \textcolor{blue}{b^{h}_{k,n}}
\end{pmatrix}
\in \M_{k,nh}(\F_q),
$$
where the block $\widetilde G_i$ consists of the $h$ columns obtained from the
$\F_q$-expansions of
$c_i,\omega c_i,\ldots,\omega^{h-1}c_i$. Equivalently, the columns of
$\widetilde G_i$ are
$\bar\Phi(c_i),\ \bar\Phi(\omega c_i),\ \ldots,\ \bar\Phi(\omega^{h-1}c_i)$.
Therefore,
$$
\widetilde G=
\left(
\begin{array}{cccc}
 \textcolor{red}{A(m_{1,1})} & \textcolor{ForestGreen}{A(m_{1,2})} &  \ldots & \textcolor{blue}{A(m_{1,n})} \\
 \textcolor{red}{A(m_{2,1})} &  \textcolor{ForestGreen}{A(m_{2,2})} & \ldots & \textcolor{blue}{A(m_{2,n})} \\
\vdots & \vdots &  & \vdots \\
\textcolor{red}{A(m_{K,1})}&  \textcolor{ForestGreen}{A(m_{K,2})}&  \ldots & \textcolor{blue}{A(m_{K,n})}
 \end{array}\right).
$$

Now, let us prove that the matrix $\widetilde G$ is an extended generator matrix for the code $\mC$. This is equivalent to show that the matrix
{\small
\[
G:=\begin{pmatrix}
\textcolor{red}{\Phi^{-1}(b^{1}_{1,1},\cdots,b^{h}_{1,1})}
& \textcolor{ForestGreen}{\Phi^{-1}(b^{1}_{1,2},\cdots,b^{h}_{1,2})}
& \cdots
& \textcolor{blue}{\Phi^{-1}(b^{1}_{1,n},\cdots,b^{h}_{1,n})}
\\
\textcolor{red}{\Phi^{-1}(b^{1}_{2,1},\cdots,b^{h}_{2,1})}
& \textcolor{ForestGreen}{\Phi^{-1}(b^{1}_{2,2},\cdots,b^{h}_{2,2})}
& \cdots
& \textcolor{blue}{\Phi^{-1}(b^{1}_{2,n},\cdots,b^{h}_{2,n})}
\\
\vdots & \vdots & & \vdots
\\
\textcolor{red}{\Phi^{-1}(b^{1}_{k,1},\cdots,b^{h}_{k,1})}
& \textcolor{ForestGreen}{\Phi^{-1}(b^{1}_{k,2},\cdots,b^{h}_{k,2})}
& \cdots
& \textcolor{blue}{\Phi^{-1}(b^{1}_{k,n},\cdots,b^{h}_{k,n})}
\end{pmatrix}
=
\left(
\begin{array}{cccc}
\textcolor{red}{g_{1,1}}
& \textcolor{ForestGreen}{g_{1,2}}
& \ldots
& \textcolor{blue}{g_{1,n}}
\\
\textcolor{red}{g_{2,1}}
& \textcolor{ForestGreen}{g_{2,2}}
& \ldots
& \textcolor{blue}{g_{2,n}}
\\
\vdots & \vdots & & \vdots
\\
\textcolor{red}{g_{K,1}}
& \textcolor{ForestGreen}{g_{K,2}}
& \ldots
& \textcolor{blue}{g_{K,n}}
\end{array}
\right)
\]
} where ${g_{i,j}^\top} \in \F_{q^h}^h$, is a generator matrix for the code $\C$ i.e. that $$\{xG : x\in \Fq^k\}=\{yM : y \in \Fqh^K\}.$$ 
Let $x=(x_1,\ldots x_K)\in \F_q^k$, where  $x_i\in \Fq^h$ for every $i\in[K]$. Then, we have
\begin{equation*} \begin{split}
xG &=(x_1,\ldots,x_K)\left(
\begin{array}{cccc}
{g_{1,1}}
& {g_{1,2}}
& \ldots
& {g_{1,n}}
\\
{g_{2,1}}
& {g_{2,2}}
& \ldots
& {g_{2,n}}
\\
\vdots & \vdots & & \vdots
\\
{g_{K,1}}
& {g_{K,2}}
& \ldots
& {g_{K,n}}
\end{array}\right) =\left(\sum_{i=1}^Kx_ig_{i,1},\ldots,\sum_{i=1}^Kx_ig_{i,n}\right).
\end{split}
\end{equation*}
Now, for each $j$-th coordinate, the following holds
\begin{equation*} \begin{split}
\sum_{i=1}^Kx_ig_{i,j}&= \sum_{i=1}^K \sum_{l=1}^hx_{i,l}\Phi^{-1}(b^{1}_{(i-1)h+l+1,j},\cdots,b^{h}_{(i-1)h+l+1,j})
\\&=\sum_{i=1}^K\Phi^{-1}(\sum_{l=1}^hx_{i,l}(b^{1}_{(i-1)h+l+1,j},\cdots,b^{h}_{(i-1)h+l+1,j}))
\\&=\sum_{i=1}^K\Phi^{-1}(x_iA(m_{i,j}))
\end{split}
\end{equation*}
Thanks to Equation~\ref{eq:propertyphi}, the following holds:  $$\forall x,y \in \Fqh, \hspace{0.5cm} \Phi^{-1}(xA(y))=\Phi^{-1}(x)y,$$ and when we apply it, we derive
\begin{equation*}
\sum_{i=1}^Kx_ig_{i,j}=\sum_{i=1}^K\Phi^{-1}(x_i)m_{i,j}
=(\Phi^{-1}(x_1),\ldots,\Phi^{-1}(x_K))c_j
=\bar{\Phi}^{-1}(x_1,\ldots,x_K)c_j.
\end{equation*}
Then, regrouping all the blocks, we obtain
\begin{equation*} \begin{split}
\{xG : x \in \Fq^k\}&=\{(\bar{\Phi}^{-1}(x_1,\ldots,x_K)c_1,\ldots,\bar{\Phi}^{-1}(x_1,\ldots,x_K)c_n) : x_1,\ldots,x_K\in \Fq^h\}
\\&=\{((y_1,\ldots,y_K)c_1,\ldots,(y_1,\ldots,y_K)c_n) : y_1,\ldots,y_K\in \Fqh\}
\\&=\{yM : y \in \Fqh^K\}.
\end{split}
\end{equation*}
Therefore, $\Tilde{G}$ is an extended generator matrix for the code $\C$ and it follows that the $i$-th coordinate subspace of this $h$-projective system associated with the code $\C$ is the projectivization of the
$\F_q$-span of the columns of $\widetilde G_i$.

We claim that this $\F_q$-span is precisely
$\bar\Phi(\langle c_i\rangle_{\F_{q^h}})$. Indeed, since
$\mathcal B=\{1,\omega,\ldots,\omega^{h-1}\}$ is an $\F_q$-basis of
$\F_{q^h}$, every $\lambda\in\F_{q^h}$ can be written uniquely as
$\lambda=\sum\limits_{r=0}^{h-1}\lambda_r\omega^r$, with $\lambda_r\in\F_q$.
Using the $\F_q$-linearity of $\bar\Phi$, we get
$$
\bar\Phi(\lambda c_i)
=
\bar\Phi\left(\sum_{r=0}^{h-1}\lambda_r\omega^r c_i\right)
=
\sum_{r=0}^{h-1}\lambda_r\bar\Phi(\omega^r c_i).
$$
Hence
$\bar\Phi(\langle c_i\rangle_{\F_{q^h}})
\subseteq
\colsp_{\F_q}(\widetilde G_i)$.
The reverse inclusion follows from the fact that each generator
$\bar\Phi(\omega^r c_i)$ belongs to
$\bar\Phi(\langle c_i\rangle_{\F_{q^h}})$. Therefore
$\colsp_{\F_q}(\widetilde G_i)
=
\bar\Phi(\langle c_i\rangle_{\F_{q^h}})$.
It follows that the $i$-th element of the $h$-projective system associated with
$\mC$ induced by the matrix $\widetilde G$ is
$$
\PG\big(\colsp_{\F_q}(\widetilde G_i)\big)
=
\PG\big(\bar\Phi(\langle c_i\rangle_{\F_{q^h}})\big)
=
\mathcal F(P_i).
$$
Therefore, an $h$-projective system associated with $\mC$ is
$\{\mathcal F(P_1),\ldots,\mathcal F(P_n)\}$.
Since each $\mathcal F(P_i)$ is an $(h-1)$-space belonging to the Desarguesian
spread obtained by field reduction, the additive code $\mC$ is faithful.
\end{proof}

We recall the following definition from \cite{alfarano2024outer}. Let
$\mathcal P=\{P_1,\ldots,P_n\}$ be a set of points in
$\PG(K-1,q^h)$. The set $\mathcal P$ is called an \textbf{outer strong blocking set} if $\bigcup_{i=1}^n \mathcal F(P_i)$
is a strong blocking set in $\PG(Kh-1,q)$.

\begin{proposition}\label{prop:outer_sbs_are_additive_sbs}
Let $h\mid k$, and write $k=Kh$. Let
$\mathcal P=\{P_1,\ldots,P_n\}\subseteq\PG(K-1,q^h)$ be an outer strong blocking set.
Then
$
\{\mathcal F(P_1),\ldots,\mathcal F(P_n)\}
$
is an additive strong blocking set in $\PG(k-1,q)$.
\end{proposition}

\begin{proof}
By definition, $\mathcal P$ is an outer strong blocking set if the union of the subspaces obtained by field reduction, namely the spaces $\mathcal F(P_i)$, is a strong blocking set in $\PG(Kh-1,q)$.
Since $k=Kh$, this means that
$
\bigcup_{i=1}^n \mathcal F(P_i)
$
is a strong blocking set in $\PG(k-1,q)$. Moreover, each $\mathcal F(P_i)$ is an
$(h-1)$-dimensional subspace of $\PG(k-1,q)$. Hence
$\{\mathcal F(P_1),\ldots,\mathcal F(P_n)\}$ is an $h$-projective system whose union
is a strong blocking set. Therefore it is an additive strong blocking set.
\end{proof}

\begin{corollary}\label{cor:outer_sbs_gives_minimal_additive}
Let $\mathcal D$ be a nondegenerate $\F_{q^h}$-linear $[n,K,d]_{q^h}$ code, and assume
that an associated projective system $\mathcal P(\mathcal D)=\{P_1,\ldots,P_n\}\subseteq\PG(K-1,q^h)$ is an outer strong
blocking set. Regard $\mathcal D$ as an additive $[n,k/h,d]^h_q$ code $\mC$, where
$k=Kh$. Then $\mC$ is a minimal additive code.
\end{corollary}
\begin{proof}
By Proposition~\ref{prop:outer_sbs_are_additive_sbs}, the collection
$\{\mathcal F(P_1),\ldots,\mathcal F(P_n)\}$ is an additive strong blocking set in
$\PG(k-1,q)$. Moreover, by Proposition~\ref{prop:linear_case_field_reduction}, it is
an $h$-projective system associated with $\mC$. Hence the conclusion follows from
Theorem~\ref{thm:minimal_additive_sbs}.
\end{proof}

\begin{remark}
The converse of Proposition~\ref{prop:outer_sbs_are_additive_sbs} does not hold in
general. An additive strong blocking set is an $h$-projective system
$\{\pi_1,\ldots,\pi_n\}$ in $\PG(k-1,q)$ such that $\bigcup_i\pi_i$ is a strong
blocking set, with only the condition $\dim(\pi_i)\leq h-1$. The subspaces $\pi_i$
need not have dimension $h-1$, and even when they do, they need not belong to a common
Desarguesian spread. Thus, additive strong blocking sets strictly generalize outer
strong blocking sets.
\end{remark}

We also recall the corresponding coding-theoretic notion. Let $\mC$ be an $\F_{q^h}$-linear code. A nonzero codeword $c\in\mC$ is called
\textbf{outer minimal} if, whenever $c'\in\mC\setminus\{0\}$ satisfies
$\supp(c')\subseteq\supp(c)$ and
$c_i'/c_i\in\F_q^*$ for every $i\in\supp(c')$, then
$c'=\lambda c$ for some $\lambda\in\F_q^*$. The code $\mC$ is called
\textbf{outer minimal} if all its nonzero codewords are outer minimal.

\begin{proposition}\label{prop:outer_minimal_additive_minimal}
Let $\mC$ be an $\F_{q^h}$-linear code. Then $\mC$ is outer minimal if and only if
$\mC$, regarded as an additive code over $\F_{q^h}$, is minimal in the sense of
Definition~\ref{def:min_codewords}.
\end{proposition}

\begin{proof}
This follows directly from the definitions. Indeed, the notion of outer minimality
for $\F_{q^h}$-linear codes requires that, whenever $c,c'\in\mC\setminus\{0\}$ satisfy
$\supp(c')\subseteq\supp(c)$ and, for every $i\in\supp(c')$, one has
$c'_i/c_i\in\F_q^*$, then $c'=\lambda c$ for some $\lambda\in\F_q^*$. This is exactly
the condition in Definition~\ref{def:min_codewords}.
\end{proof}

Thus, the notions introduced in this paper generalize both the classical theory of strong blocking sets and minimal linear codes, and the theory of outer strong blocking sets and outer minimal codes. The case $h=1$ recovers the classical correspondence between minimal $\F_q$-linear codes and strong blocking sets. On the other hand, when $h\mid k$ and the associated
$(h-1)$-spaces lie in a Desarguesian spread, we recover the theory of outer strong blocking sets and outer minimal codes over $\F_{q^h}$. Hence additive strong blocking sets and minimal additive codes provide a common framework containing both settings.

%%%%%%%%%%%%%%%%%%%%%%%%%%%%%%%%%%%%%%%%%5

\section{Existence results and bounds for minimal additive codes}\label{sec:constructions_bounds} 

In this section, we study existence results and bounds on the parameters of minimal additive codes or, equivalently, on the size of additive strong blocking sets. We first give a counting argument which ensures the existence of minimal additive codes for suitable parameters. We then introduce the function $m(k,q,h)$, the smallest length of a nondegenerate  minimal additive
code of $\Fq$-dimension $k$ over $\Fqh$, and derive upper and lower bounds for this parameter. Some of the bounds obtained in this section are stated in a
more general form, for arbitrary minimal additive codes, and then specialized to $m(k,q,h)$.

\subsection{A combinatorial existence result}

We begin with a counting argument, in the spirit of the classical existence results for minimal codes and of the analogous result for outer minimal codes in \cite[Theorem 38]{alfarano2024outer}.

\begin{theorem}\label{thm:existence}
    If
    \[\binom{nh-2}{k-2}_{q}\cdot \sum_{i=1}^n \binom{n}{i}(q^h-1)^i\left(q^i-q\right)<\binom{nh}{k}_{q},\]
    then there exists an $[n,k/h]_q^h$ minimal  additive code. 
\end{theorem}
\begin{proof}
An $[n,k/h]_q^h$ additive code $\mC$ is not minimal if it contains a codeword which is not minimal. A nonzero codeword $c\in \mC\subseteq  \F_{q^h}^n$ is not minimal if there exists $c'\in \mC$ such that it holds
\begin{equation}\label{eq:not-min}
\left\{
\begin{aligned}
& c'\neq \lambda c, \ \forall \lambda\in \F_q \ \ \wedge \\
& \sH(c')\subseteq \sH(c) \ \ \wedge \\
& \forall i\in \sH(c), \ \exists \lambda_i\in \F_q
\text{ s.t. } c'_i=\lambda_i c_i.
\end{aligned}
\right.
\end{equation}
Let 
\[B=\left\{(c,c')\in \F_{q^h}^n \times \F_{q^h}^n : c\neq 0 \text{ and }\eqref{eq:not-min} \text{ holds}\right\}.\]
The cardinality of $B$ is given by
    $$|B|=\sum_{i=1}^n \binom{n}{i}(q^h-1)^i\left(q^i-q\right).$$
An $[n,k/h]_q^h$ additive code is not minimal if it contains the $\F_q$-subspace generated by one of the pairs in $B$. Since the $2$-dimensional $\F_q$-subspace generated by each pair in $B$ is contained in exactly $\binom{nh-2}{k-2}_{q}$ additive codes of dimension~$k$ in $\F_{q^h}^n$ and the total number of $[n,k/h]_q^h$ codes is $\binom{nh}{k}_{q}$, we have that if $\binom{nh-2}{k-2}_{q}\cdot |B| < \binom{nh}{k}_{q}$, there exists an $[n,k/h]_q^h$ additive code $\mC$ which does not contain any $2$-dimensional $\F_q$ subspace generated by a pair in $B$. Therefore, such $\mC$ must be minimal.
\end{proof}

 \begin{corollary}\label{cor:existence}
     $[n,k/h]_q^h$  minimal additive codes exist whenever
    \[n\geq \left\lceil \frac{2k}{h\log_{q^h}\left(\frac{q^{2h}}{q^{h+1}-q+1}\right)}\right\rceil. \]
 \end{corollary}
\begin{proof}
    First of all observe that 
    \[|B|< (1+(q^h-1)q)^n\]
and
\[\frac{\binom{nh}{k}_{q}}{\binom{nh-2}{k-2}_{q}}=\frac{(q^{nh}-1)(q^{nh-1}-1)}{(q^{k}-1)(q^{k-1}-1)}\geq q^ {2nh-2k}.
\]
In order to ensure the existence of minimal additive codes, from Theorem~\ref{thm:existence}, it is enough to have 
\[(1+(q^h-1)q)^n < q^ {2nh-2k}.\]
Taking logarithms base $q^h$ yields
\[
n \log_{q^h}(q^{h+1}-q+1) < 2\left(n-\frac{k}{h}\right).
\]
Rearranging, we obtain
\[
n\bigl(2-\log_{q^h}(q^{h+1}-q+1)\bigr) > \frac{2k}{h}.
\]
Equivalently,
\[
n \ge
\left\lceil
\frac{2k}{h\,\log_{q^h}\!\left(\frac{q^{2h}}{q^{h+1}-q+1}\right)}
\right\rceil.
\]
\end{proof}

\begin{remark}
Theorem~\ref{thm:existence} is the additive version of
\cite[Theorem~38]{alfarano2024outer}. In the latter case one counts $K$-dimensional $\F_{q^h}$-linear subspaces of $\F_{q^h}^N$, and the corresponding
condition is
$$
\binom{N-2}{K-2}_{q^h}
\sum_{i=1}^N
\binom{N}{i}(q^h-1)^i(q^i-q)
<
\binom{N}{K}_{q^h}.
$$
In the additive setting we count instead $k$-dimensional $\F_q$-subspaces of $\F_{q^h}^n$.
The set of bad pairs is the same, and the difference lies only in the ambient family of codes being
counted. When $h\mid k$ and $k=Kh$, Corollary~\ref{cor:existence}
recovers the bound of \cite[Theorem~38]{alfarano2024outer}. Thus this existence result for minimal additive codes extends the
corresponding result for outer minimal codes.
\end{remark}

\subsection{Bounds on the length of minimal additive codes and first constructions}

In this subsection we investigate the smallest length that a minimal additive code can have for fixed field size and dimension. To this end, we define the parameter
$$m(k,q,h):= \min \{n \in \mathbb{N}_{\geq 1} : \text{ there exists an } [n,k/h]_q^h \text{ nondegenerate minimal additive code}\}.$$
By Corollary~\ref{cor:existence}, this number is well-defined for every choice of $k,q,h$. The aim of this subsection is to provide lower and upper bounds on $m(k,q,h)$.

We first observe that the existence of one minimal additive code of length $m(k,q,h)$ implies the existence of minimal additive codes of any larger length. 

\begin{proposition}\label{thm:monotonicity}
There exists a nondegenerate $[n,k/h]_q^h$ minimal additive code if and only if $n\geq m(k,q,h)$.
\end{proposition}

\begin{proof}
If $n<m(k,q,h)$, then no $[n,k/h]_q^h$ minimal additive code exists by definition of $m(k,q,h)$.

Conversely, suppose that $n\geq m(k,q,h)$, and let $\mC$ be a $[m(k,q,h),k/h]_q^h$ minimal additive code. Let $\chi(\mC)=\{\pi_1,\ldots,\pi_{m(k,q,h)}\}$ be an associated $h$-projective system. By Theorem~\ref{thm:minimal_additive_sbs}, $\chi(\mC)$ is an additive strong blocking set, that is, $\bigcup\limits_i\pi_i$ is a strong blocking set in $\PG(k-1,q)$.
We obtain an $h$-projective system of length $n$ by adding $n-m(k,q,h)$ repetitions of some of the subspaces $\pi_i$. Since repetitions do not change the union $\bigcup_i\pi_i$, the new projective system is still an additive strong blocking set. Hence, again by Theorem~\ref{thm:minimal_additive_sbs}, the corresponding additive code is minimal. Therefore, an $[n,k/h]_q^h$ minimal additive code exists for every $n\geq m(k,q,h)$.
\end{proof}

As a direct consequence of Corollary~\ref{cor:existence}, we obtain the following upper bound on $m(k,q,h)$.

\begin{theorem}\label{thm:upper-counting}
For every prime power $q$ and all positive integers $k,h$, we have
$$
m(k,q,h) \leq
\left\lceil
\frac{2k}{h\,\log_{q^h}\!\left(\frac{q^{2h}}{q^{h+1}-q+1}\right)}
\right\rceil.
$$
\end{theorem}

We now explain how known geometric constructions of strong blocking sets produce minimal additive codes. Several constructions of strong blocking
sets in projective spaces are obtained as unions of subspaces, in particular union of lines; see \cite{heger2021short, alfarano2022three, alon2024strong, fancsali2014lines}.
The next result shows that any such construction immediately yields a
minimal additive code, provided that the subspaces have dimension at most
$h-1$.

\begin{proposition}\label{prop:upper-subspaces}
Let $\pi_1,\ldots,\pi_m$ be subspaces of $\PG(k-1,q)$ such that
$\dim(\pi_i)\leq h-1$ for every $i\in[m]$. If
$\bigcup\limits_{i=1}^m \pi_i$ is a strong blocking set in $\PG(k-1,q)$, then
$\{\pi_1,\ldots,\pi_m\}$ is an additive strong blocking set in $\PG(k-1,q)$. In particular, there exists a nondegenerate $[m,k/h]^h_q$ minimal additive code, and $m(k,q,h)\leq m$.
\end{proposition}

\begin{proof}
Since $\bigcup\limits_{i=1}^m \pi_i$ is a strong blocking set, it is not contained
in a hyperplane of $\PG(k-1,q)$. Hence $\{\pi_1,\ldots,\pi_m\}$ is an
$h$-projective system. Moreover, its union is a strong blocking set by
assumption, so it is an additive strong blocking set. Theorem~\ref{thm:minimal_additive_sbs}
then gives a nondegenerate $[m,k/h]^h_q$ minimal additive code. Therefore
$m(k,q,h)\leq m$.
\end{proof}

We obtain the following immediate consequence.

\begin{corollary}\label{cor:upper-linear}
For every $k,q,h$, we have $m(k,q,h)\leq m(k,q)$, where $m(k,q)$ denotes the smallest length of a $[n,k]_q$ minimal linear code.
\end{corollary}

\begin{proof}
Let $n=m(k,q)$ and let $\mC$ be a minimal $[n,k]_q$ linear code. By the classical correspondence between minimal linear codes and strong blocking sets, the associated projective system $\mathcal{P}(\mC)$ is a strong blocking set in $\PG(k-1,q)$. Since points are subspaces of dimension $0\leq h-1$, Proposition~\ref{prop:upper-subspaces} gives the existence of an $[n,k/h]^h_q$ minimal additive code. Hence $m(k,q,h)\leq n=m(k,q)$.
\end{proof}

Combining Corollary~\ref{cor:upper-linear} with the best-known general upper bounds for $m(k,q)$ gives explicit upper bounds on $m(k,q,h)$. In particular, by
\cite[Remark 4.11]{alfarano2024outer}, we have that
$$
m(k,q,h)\le m(k,q)\le
\left\lceil
\frac{k}{\log_{q^2}\left(\frac{q^4}{q^3-q+1}\right)}
\right\rceil (q+1).
$$
We point out that the same upper bound for the minimal length of a minimal linear code has been also independently shown in \cite{heger2021short,bishnoi2024blocking}.

\begin{example}\label{ex:tetrahedron}
Let $P_1,P_2,P_3,P_4$ be four points in general position in $\PG(3,q)$ and,
for $1\leq i<j\leq 4$, let
$\ell_{ij}:=\langle P_i,P_j\rangle_{\Fq}$ be the line joining $P_i$ and
$P_j$. The union of the six lines $\ell_{ij}$ is a strong blocking set in
$\PG(3,q)$.
First, consider $\chi_2:=\{\ell_{ij}:1\leq i<j\leq 4\}$.
This is a $2$-projective system of length $6$ in $\PG(3,q)$, and its union
is a strong blocking set. Hence $\chi_2$ is an additive strong blocking set.
By Theorem~\ref{thm:minimal_additive_sbs}, it gives rise to an $[6,4/2,3]^2_q$ minimal additive code. 
The corresponding linear code is a $[6q-2,4,3q-2]_q$ minimal code. 

If $h=3$, we can further group the edges inside faces of the tetrahedron.
For instance, let
$$
\pi_1:=\langle P_2,P_3,P_4\rangle_{\Fq},\qquad
\pi_2:=\langle P_1,P_3,P_4\rangle_{\Fq},\qquad
\pi_3:=\langle P_1,P_2,P_4\rangle_{\Fq}.
$$
Then $\chi_3:=\{\pi_1,\pi_2,\pi_3\}$ is a $3$-projective system of length $3$ in $\PG(3,q)$. Moreover,
$\pi_1\cup\pi_2\cup\pi_3$ contains all six edges of the tetrahedron. Since
the union of these six edges is a strong blocking set, also
$\pi_1\cup\pi_2\cup\pi_3$ is a strong blocking set. Therefore $\chi_3$ is an
additive strong blocking set, and it gives rise to a $[3,4/3,2]^3_q$ minimal additive code. Here the minimum distance is $2$, since at
most one of the three planes $\pi_1,\pi_2,\pi_3$ can be contained in a
hyperplane of $\PG(3,q)$; see
Figure~\ref{fig:tetrahedron-additive-sbs} for an illustration of the two different constructions.
This illustrates the trade-off
offered by the additive framework,  by allowing larger alphabets, one can group several points of the same strong blocking set into a single subspace, thereby reducing the length of the code. The price to pay
is that the resulting code is no longer $\F_q$-linear, but only additive over $\F_{q^h}$.

\begin{figure}[ht]
\centering
\begin{tikzpicture}[scale=0.9, line cap=round, line join=round]

% First tetrahedron: six edges
\begin{scope}[xshift=0cm]
\coordinate (P1) at (0,0);
\coordinate (P2) at (3,0);
\coordinate (P3) at (1.15,1.15);
\coordinate (P4) at (1.45,3);

\draw[very thick, red] (P1)--(P2);
\draw[very thick, blue] (P1)--(P3);
\draw[very thick, orange] (P1)--(P4);
\draw[very thick, green!60!black] (P2)--(P3);
\draw[very thick, purple] (P2)--(P4);
\draw[very thick, cyan!70!black] (P3)--(P4);

\fill (P1) circle (2pt) node[below left] {$P_1$};
\fill (P2) circle (2pt) node[below right] {$P_2$};
\fill (P3) circle (2pt) node[right] {$P_3$};
\fill (P4) circle (2pt) node[above] {$P_4$};

\node[red] at (1.5,-0.25) {$\ell_{12}$};
\node[blue] at (0.42,0.75) {$\ell_{13}$};
\node[orange] at (0.52,1.65) {$\ell_{14}$};
\node[green!60!black] at (2.35,0.75) {$\ell_{23}$};
\node[purple] at (2.45,1.65) {$\ell_{24}$};
\node[cyan!70!black] at (1.05,2.1) {$\ell_{34}$};

\node at (1.5,-0.8) {$\chi_2$: six lines};
\end{scope}

% Second tetrahedron: three faces
\begin{scope}[xshift=6cm]
\coordinate (Q1) at (0,0);
\coordinate (Q2) at (3,0);
\coordinate (Q3) at (1.15,1.15);
\coordinate (Q4) at (1.45,3);

\fill[red!35, opacity=0.55] (Q2)--(Q3)--(Q4)--cycle;
\fill[blue!35, opacity=0.45] (Q1)--(Q3)--(Q4)--cycle;
\fill[green!35, opacity=0.45] (Q1)--(Q2)--(Q4)--cycle;

\draw[thick] (Q1)--(Q2);
\draw[thick] (Q1)--(Q3);
\draw[thick] (Q1)--(Q4);
\draw[thick] (Q2)--(Q3);
\draw[thick] (Q2)--(Q4);
\draw[thick] (Q3)--(Q4);

\fill (Q1) circle (2pt) node[below left] {$P_1$};
\fill (Q2) circle (2pt) node[below right] {$P_2$};
\fill (Q3) circle (2pt) node[right] {$P_3$};
\fill (Q4) circle (2pt) node[above] {$P_4$};

\node[red!70!black] at (2.15,1.45) {$\pi_1$};
\node[blue!70!black] at (0.55,1.55) {$\pi_2$};
\node[green!50!black] at (1.55,0.52) {$\pi_3$};

\node at (1.5,-0.8) {$\chi_3$: three planes};
\end{scope}

\end{tikzpicture}
\caption{Two additive strong blocking sets in $\PG(3,q)$ obtained from the tetrahedron. On the left, the six edges give a $2$-projective system of length $6$; on the right, three faces give a $3$-projective system of length $3$.}
\label{fig:tetrahedron-additive-sbs}
\end{figure}
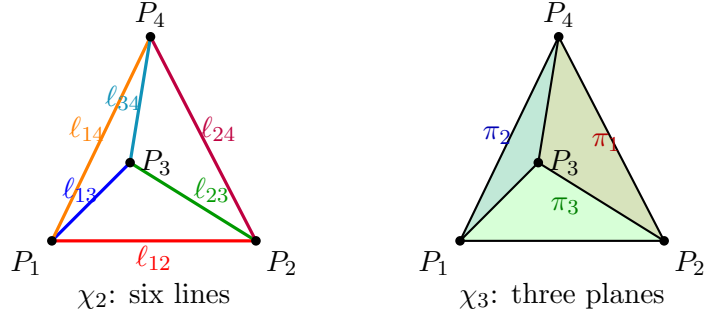
\end{example}

We now derive some lower bounds on $m(k,q,h)$. We start with the trivial dimension bound. Since an additive code of length $n$ over $\Fqh$ has $\Fq$-dimension at most $nh$, we have
$$
m(k,q,h)\geq \left\lceil \frac{k}{h}\right\rceil.
$$

The next lower bound follows from the geometric characterization of minimal additive codes. We will use the Jamison--Brouwer--Schrijver bound \cite{jamison1977covering, brouwer1978blocking} for affine blocking sets, which has been exploited to show that every strong blocking set in $\PG(k-1,q)$ has size at least $(k-1)(q+1)$; see \cite{alfarano2022three, heger2021short}.

\begin{theorem}\label{thm:lower-geometric}
Let $\chi=\{\pi_1,\ldots,\pi_n\}$ be an additive strong blocking set in
$\PG(k-1,q)$. For every $i\in[n]$, let $t_i$ be the vectorial dimension of
$\pi_i$. Then
$$
\sum_{i=1}^n \frac{q^{t_i}-1}{q-1}\geq (k-1)(q+1).
$$
In particular, if $\chi$ has length $n$ and $\dim(\pi_i)\leq h-1$ for every
$i\in[n]$, then
$$
n\geq
\left\lceil
\frac{(q^2-1)(k-1)}{q^h-1}
\right\rceil .
$$
\end{theorem}

\begin{proof}
By Theorem~\ref{thm:minimal_additive_sbs}, $\chi$ is an additive strong blocking set. Hence
$S:=\bigcup\limits_{i=1}^n \pi_i$ is a strong blocking set in $\PG(k-1,q)$. By the Jamison--Brouwer--Schrijver bound \cite{jamison1977covering, brouwer1978blocking} and \cite[Theorem 2.14]{alfarano2022three}, we have $|S|\geq (k-1)(q+1)$.
On the other hand, since $\pi_i$ has vectorial dimension $t_i$, it contains exactly
$(q^{t_i}-1)/(q-1)$ points. Therefore,
$$
|S|\leq \sum_{i=1}^n |\pi_i|
=
\sum_{i=1}^n \frac{q^{t_i}-1}{q-1}.
$$
This implies that
$$
\sum_{i=1}^n \frac{q^{t_i}-1}{q-1}\geq (k-1)(q+1).
$$
Since $t_i\leq h$ for every $i$, we have
$$
\frac{q^{t_i}-1}{q-1}\leq \frac{q^h-1}{q-1}.
$$
Hence
$$
(k-1)(q+1)\leq n\frac{q^h-1}{q-1},
$$
and so
$$
n\geq \frac{(q^2-1)(k-1)}{q^h-1}.
$$
Taking $n=m(k,q,h)$ gives
$$
m(k,q,h)\geq
\left\lceil
\frac{(q^2-1)(k-1)}{q^h-1}
\right\rceil .
$$
\end{proof}

\begin{corollary}\label{cor:combined-lower-bound}
For every prime power $q$ and all positive integers $k,h$, we have
$$
m(k,q,h)\geq
\max\left\{
\left\lceil \frac{k}{h}\right\rceil,
\left\lceil
\frac{(q^2-1)(k-1)}{q^h-1}
\right\rceil
\right\}.
$$
\end{corollary}

\begin{remark}
The first inequality in Theorem~\ref{thm:lower-geometric} is stronger than the final lower bound on $m(k,q,h)$, since it takes into account the actual dimensions of the subspaces $\pi_i$. In particular, if many of the $\pi_i$ have vectorial dimension strictly smaller than $h$, then
$$
\sum_{i=1}^n \frac{q^{t_i}-1}{q-1}\geq (k-1)(q+1)
$$
may force a larger value of $n$ than the bound
$$
n\geq
\left\lceil
\frac{(q^2-1)(k-1)}{q^h-1}
\right\rceil .
$$
For $h=1$, the bound becomes $n\geq (k-1)(q+1)$, which is the classical lower bound for minimal linear codes.
\end{remark}

\subsection{A lower bound on the minimum distance}

We now derive a lower bound on the minimum distance of a minimal additive code.
This is again obtained from the geometric characterization in terms of additive
strong blocking sets, together with the affine form Jamison--Brouwer--Schrijver bound.

\begin{theorem}\label{thm:distance-lower-bound}
Let $\mC$ be a $[n,k/h,d]^h_q$ nondegenerate minimal additive code, and let
$\chi(\mC)=\{\pi_1,\ldots,\pi_n\}$ be its associated $h$-projective system in
$\PG(k-1,q)$. For every $i\in[n]$, let $t_i$ be the vectorial
dimension of $\pi_i$. Then, for every hyperplane $H$ of $\PG(k-1,q)$, one has
$$
\sum_{\pi_i\not\subseteq H} q^{t_i-1}\geq (k-1)(q-1)+1.
$$
In particular,
$$
d\geq
\left\lceil
\frac{(k-1)(q-1)+1}{q^{h-1}}
\right\rceil .
$$
\end{theorem}

\begin{proof}
By Theorem~\ref{thm:minimal_additive_sbs}, $\chi(\mC)$ is an additive strong
blocking set. Hence
$$
S:=\bigcup_{i=1}^n \pi_i
$$
is a strong blocking set in $\PG(k-1,q)$.
Let $H$ be a hyperplane of $\PG(k-1,q)$. Since $S$ is a strong blocking set,
the set $S\setminus H$ is a blocking set with respect to the affine hyperplanes
of $\PG(k-1,q)\setminus H\simeq \mathrm{AG}(k-1,q)$; see, for instance, \cite{alon2024strong}.
By the affine Jamison--Brouwer--Schrijver bound, we get
$$
|S\setminus H|\geq (k-1)(q-1)+1.
$$
On the other hand,
$$
S\setminus H\subseteq \bigcup_{\pi_i\not\subseteq H}(\pi_i\setminus H).
$$
If $\pi_i\not\subseteq H$ and $\pi_i$ has vectorial dimension $t_i$, then
$\pi_i\cap H$ is a hyperplane of $\pi_i$. Hence $\pi_i\setminus H$ has exactly
$q^{t_i-1}$ points. Therefore
$$
|S\setminus H|
\leq
\sum_{\pi_i\not\subseteq H} q^{t_i-1}.
$$
This proves
$$
\sum_{\pi_i\not\subseteq H} q^{t_i-1}\geq (k-1)(q-1)+1.
$$
Now choose a hyperplane $H$ such that
$$
d=n-|\{i:\pi_i\subseteq H\}|.
$$
Then the number of indices $i$ such that $\pi_i\not\subseteq H$ is exactly
$d$. Since $t_i\leq h$ for every $i$, we obtain
$$
(k-1)(q-1)+1
\leq
\sum_{\pi_i\not\subseteq H} q^{t_i-1}
\leq
dq^{h-1}.
$$
Hence
$$
d\geq
\left\lceil
\frac{(k-1)(q-1)+1}{q^{h-1}}
\right\rceil .
$$
\end{proof}

\subsection{Asymptotic behaviour}
We study the asymptotic behaviour of $m(k,q,h)$ as $k\to+\infty$.

For an $[n,k/h,d]_q^h$ additive code, we denote by $R:=\frac{k}{n}$ the information rate and by $\delta:=\frac{d}{n}$ the relative minimum distance.
\begin{lemma}\label{lem:minimality_constraint}
    Let $\C$ be an $[n,k/h,d]_q^h$ minimal additive code. Then $d \geq \frac{k}{h}$ or, equivalently, $R \leq \delta h$.
\end{lemma}

\begin{proof}
    Take $c\in \C$ of minimal weight, so that $\wt(c)=d$. Puncture $\mC$ on the $n-d$ coordinates in which $c$ is zero, or equivalently, keep only the coordinates in $\supp(c)$. Denote the punctured code by $\mC'$. We claim that $\mC'$ is an $[d,k/h,d']_q^h$ additive code. Suppose, by contradiction, that $\dim_{\F_q}(\mC')<k$. Then there exists a nonzero codeword $x\in \mC$ which is zero on all coordinates of $\supp(c)$. Hence $
\supp(c)\subsetneq \supp(c+x)$ and  $\C$ is not  minimal. Now, by the Singleton bound applied on $\C'$, we obtain $$\frac{k}{h}\leq d-d'+1.$$ Since $d'\geq 1$, we get $d \geq \frac{k}{h}$.
\end{proof}

We now combine this with the classical asymptotic Plotkin bound. We apply it by viewing
an additive code $\mC\subseteq \F_{q^h}^n$ simply as a possibly nonlinear code over the
alphabet $\F_{q^h}$. Thus the alphabet size is $q^h$, and the size of the code is
$|\mC|=q^k$. This is the same bound that we get in Remark~\ref{rem:faithful_plotkin}.

\begin{theorem}\label{thm:asympt_rate}
We have
$$
\liminf_{k\to+\infty}\frac{m(k,q,h)}{k}
\geq
\frac{2q^h-1}{h(q^h-1)}.
$$
\end{theorem}

\begin{proof}
Let $(\mC_m)_m$ be a sequence of minimal additive codes with parameters
$[n_m,k_m/h,d_m]^h_q$, where $k_m\to+\infty$. Set $R_m:=k_m/n_m$ and
$\delta_m:=d_m/n_m$. By the previous lemma, minimality implies $R_m\leq h\delta_m$.
We now regard $\mC_m$ simply as a possibly nonlinear code in $\F_{q^h}^{n_m}$. Its
alphabet has size $q^h$ and its cardinality is $|\mC_m|=q^{k_m}$. Hence the classical
asymptotic Plotkin bound gives
$$
R_m\leq h-\frac{h\delta_m}{1-\frac{1}{q^h}}+o(1)
$$
whenever $0\leq \delta_m<1-\frac{1}{q^h}$, while $R_m=o(1)$ if
$\delta_m>1-\frac{1}{q^h}$.

Thus, asymptotically, the rate of a minimal additive code is bounded above by the
minimum of the two functions $h\delta$ and
$h-\frac{h\delta}{1-\frac{1}{q^h}}$. This is illustrated in
Figure~\ref{fig:asymptotic_bounds_minimal_additive}.

The largest possible value of this minimum is attained when the two functions are equal.
Thus we solve
$$
h\delta=
h-\frac{h\delta}{1-\frac{1}{q^h}},
$$
which gives
$\delta=\frac{1-\frac{1}{q^h}}{2-\frac{1}{q^h}}
=\frac{q^h-1}{2q^h-1}$. Hence the corresponding value of the rate is
$R=h\delta=h\frac{q^h-1}{2q^h-1}$. Therefore every asymptotic family of minimal
additive codes satisfies
$$
\limsup_{m\to+\infty}R_m
\leq
h\frac{q^h-1}{2q^h-1}.
$$
Equivalently,
$$
\liminf_{m\to+\infty}\frac{n_m}{k_m}
\geq
\frac{2q^h-1}{h(q^h-1)}.
$$
Taking $n_m=m(k_m,q,h)$ gives the desired lower bound on
$\liminf_{k\to+\infty}m(k,q,h)/k$.

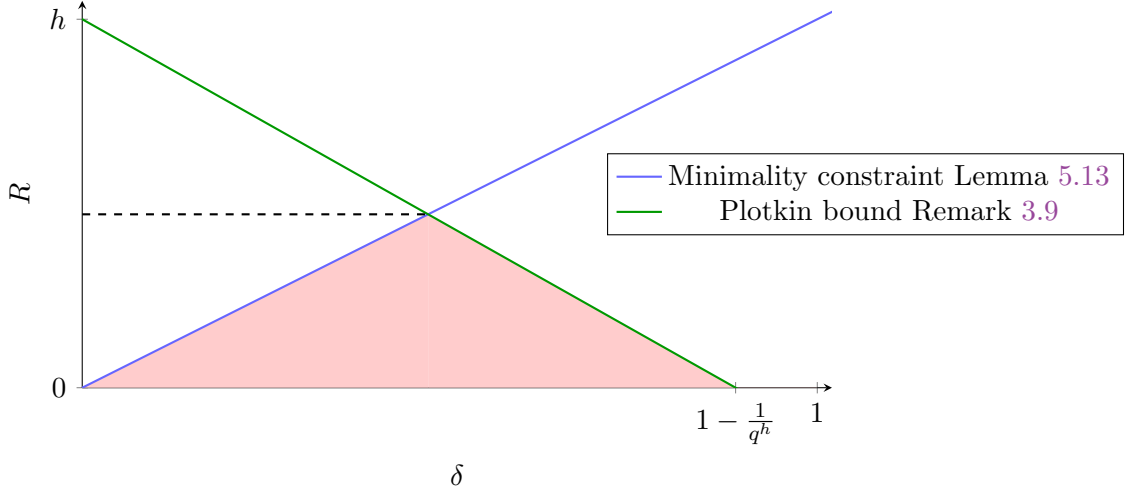
\begin{figure}[ht]
\centering
\begin{tikzpicture}
\begin{axis}[
    width=11.5cm,
    height=6.7cm,
    axis lines=left,
    xmin=0, xmax=1.02,
    ymin=0, ymax=2.1,
    xtick={0.8889,1},
    xticklabels={$1-\frac{1}{q^h}$,$1$},
    ytick={0,2},
    yticklabels={$0$,$h$},
    xlabel={$\delta$},
    ylabel={$R$},
    legend style={at={(0.7,0.5)},anchor=west},
]

% Intersection
\def\xint{0.4706}

% Zone remplie (gauche)
\addplot[fill=red!20, draw=none, domain=0:\xint, forget plot] {2*x} \closedcycle;

% Zone remplie (droite)
\addplot[fill=red!20, draw=none, domain=\xint:1, forget plot] {2 - 2.25*x} \closedcycle;

% Minimality constraint
\addplot[blue!60, thick] {2*x};
\addlegendentry{Minimality constraint Lemma~\ref{lem:minimality_constraint}}

% Plotkin-type bound
\addplot[green!60!black, thick] {2 - 2.25*x};
\addlegendentry{Plotkin bound Remark~\ref{rem:faithful_plotkin}}

% Coordonnées intersection
\def\yint{0.9412}

% Ligne horizontale pointillée
\addplot[dashed, thick, forget plot] coordinates {(0,\yint) (\xint,\yint)};

\end{axis}
\end{tikzpicture}
\caption{Comparison of the minimality constraint with the Plotkin-type bound. The green line denotes the classical Plotkin bound whereas the blue line corresponds to the bound arising from the minimality condition. The shaded red region represents the allowed $(\delta, R)$-region determined by the intersection of these two constraints.}
\label{fig:asymptotic_bounds_minimal_additive}
\end{figure}

\end{proof}

\begin{remark}
The lower bound from Theorem~\ref{thm:lower-geometric}
is not asymptotically tight for $h\geq 2$. Indeed, Theorem~\ref{thm:asympt_rate}
gives
$$
\liminf_{k\to+\infty}\frac{m(k,q,h)}{k}
\geq
\frac{2q^h-1}{h(q^h-1)}.
$$
On the other hand, Theorem~\ref{thm:lower-geometric} gives
$$
m(k,q,h)\geq
\left\lceil
\frac{(q^2-1)(k-1)}{q^h-1}
\right\rceil,
$$
whose asymptotic coefficient is
$$
\frac{q^2-1}{q^h-1}.
$$
For every $h\geq 2$, one has
$$
\frac{2q^h-1}{h(q^h-1)}
>
\frac{q^2-1}{q^h-1}.
$$
Indeed, this inequality is equivalent to
$$
2q^h-1>h(q^2-1).
$$
For $h=2$, this becomes $2q^2-1>2q^2-2$, which is clear. If $h\geq 3$, then
$q\geq 2$ gives $2q^{h-2}\geq 2^{h-1}\geq h$, and hence $2q^h\geq hq^2$.
Therefore,
$$
2q^h-1-h(q^2-1)=2q^h-hq^2+h-1\geq h-1>0.
$$
Thus the asymptotic lower bound in Theorem~\ref{thm:asympt_rate}
is strictly stronger than the bound in Theorem~\ref{thm:lower-geometric} for
every $h\geq 2$. In particular, the bound in Theorem~\ref{thm:lower-geometric}
cannot be asymptotically tight in the additive setting when $h\geq 2$.
\end{remark}

\begin{remark}
The Plotkin bound used in the proof of Theorem~\ref{thm:asympt_rate} is the classical Plotkin bound
for possibly nonlinear codes over the alphabet $\F_{q^h}$, which we have obtained in Remark~\ref{rem:faithful_plotkin} for faithful codes. Indeed, in that case
$\nu_h=1$ and $\nu_j=0$ for every $j<h$, so $\beta=1-1/q^h$ and, for
$0\leq\delta<\beta$, one has
$M_\nu(\delta)=h\delta/(1-1/q^h)$. Hence Theorem~\ref{thm:asymptotic_plotkin_ti}
gives
$$
R\leq h-M_\nu(\delta)=h-\frac{h\delta}{1-\frac{1}{q^h}},
$$
which is precisely the asymptotic Plotkin bound used above.
For non-faithful additive codes, Theorem~\ref{thm:asymptotic_plotkin_ti} may give a
sharper estimate. 
Indeed, if a sequence of minimal additive codes has  dimension distribution $\nu=(\nu_1,\ldots,\nu_h)$, then the refined Plotkin bound gives
$R\leq \sum_{j=1}^h j\nu_j-M_\nu(\delta)$ whenever $\delta<\beta$. Combining this with
the constraint $R\leq h\delta$ gives an asymptotic rate bound depending on $\nu$. In particular, for families with many coordinates of dimension strictly smaller
than $h$, this may improve the lower bound on the asymptotic length. We illustrate this in the following example.
\end{remark}

\begin{example}
Let $q=2$ and $h=2$, and consider a sequence of additive codes whose dimension
distribution is given by $\nu_1=\nu_2=1/2$. Then
$\beta=\frac12(1-\frac12)+\frac12(1-\frac14)=5/8$. From the definition of
$M_\nu(\delta)$, one obtains
$$
M_\nu(\delta)=
\begin{cases}
\dfrac{8}{3}\delta, & 0\leq \delta\leq \dfrac38,\\[4pt]
2\delta+\dfrac14, & \dfrac38\leq \delta\leq \dfrac58.
\end{cases}
$$
Indeed, $M_\nu(\delta)$ is obtained by maximizing $y_1+2y_2$ under the constraints
$0\leq y_1,y_2\leq 1/2$ and $\frac12y_1+\frac34y_2\leq\delta$. For small $\delta$, the
maximum is achieved by taking $y_2$ as large as possible, giving
$M_\nu(\delta)=8\delta/3$ until $y_2=1/2$, namely until $\delta=3/8$. After that, one has
$y_2=1/2$, and the remaining contribution comes from increasing $y_1$, giving
$M_\nu(\delta)=2\delta+1/4$.
Since $\sum_{j=1}^2 j\nu_j=3/2$, Theorem~\ref{thm:asymptotic_plotkin_ti} gives
$$
R\leq
\begin{cases}
\dfrac32-\dfrac{8}{3}\delta, & 0\leq \delta\leq \dfrac38,\\[4pt]
\dfrac54-2\delta, & \dfrac38\leq \delta\leq \dfrac58.
\end{cases}
$$
Together with the minimality constraint $R\leq 2\delta$, the intersection is given by
$2\delta=3/2-8\delta/3$, hence $\delta=9/28$ and $R=9/14$. Thus, for such a
non-faithful family, one obtains the asymptotic constraint $n/k\geq14/9$, which is
stronger than the constraint obtained from the classical Plotkin bound over $\F_4$. This
comparison is illustrated in Figure~\ref{fig:nonfaithful_refined_plotkin}.

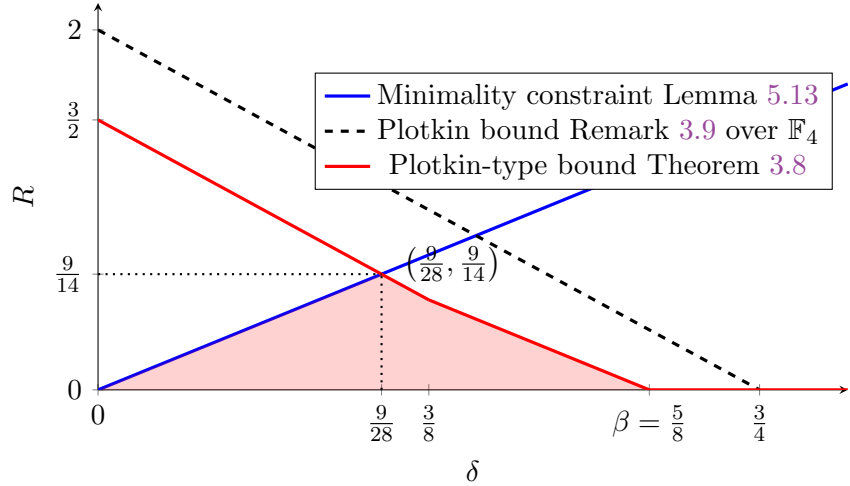
\begin{figure}[ht]
\centering
\begin{tikzpicture}
\begin{axis}[
    width=11.5cm,
    height=6.7cm,
    xmin=0, xmax=0.85,
    ymin=0, ymax=2.15,
    axis lines=left,
    xlabel={$\delta$},
    ylabel={$R$},
    xtick={0,0.3214285714,0.375,0.625,0.75},
    xticklabels={$0$,$\frac{9}{28}$,$\frac38$,$\beta=\frac58$,$\frac34$},
    ytick={0,0.6428571428,1.5,2},
    yticklabels={$0$,$\frac{9}{14}$,$\frac32$,$2$},
    legend style={at={(0.98,0.82)},anchor=north east},
    samples=200,
]

% Minimality bound R = 2 delta
\addplot[very thick, blue, domain=0:0.85] {2*x};
\addlegendentry{Minimality constraint Lemma~\ref{lem:minimality_constraint}}

% Classical Plotkin bound over F_4
\addplot[very thick, dashed, black, domain=0:0.75] {2-(8/3)*x};
\addlegendentry{Plotkin bound Remark~\ref{rem:faithful_plotkin} over $\F_4$}

% Refined Plotkin bound for nu_1=nu_2=1/2
\addplot[very thick, red, domain=0:0.375] {1.5-(8/3)*x};
\addlegendentry{Plotkin-type bound Theorem~\ref{thm:asymptotic_plotkin_ti}}

\addplot[very thick, red, domain=0.375:0.625] {1.25-2*x};

% Zero part after beta
\addplot[very thick, red, domain=0.625:0.85] {0};

% Admissible region below both minimality and refined Plotkin
\addplot[
    fill=red,
    fill opacity=0.18,
    draw=none
]
coordinates {
    (0,0)
    (0.3214285714,0.6428571428)
    (0.375,0.5)
    (0.625,0)
    (0,0)
};

% Intersection of minimality and refined Plotkin
\addplot[dotted, thick, black] coordinates {(0.3214285714,0) (0.3214285714,0.6428571428)};
\addplot[dotted, thick, black] coordinates {(0,0.6428571428) (0.3214285714,0.6428571428)};
\node at (axis cs:0.335,0.70) [anchor=west]
{$\left(\frac{9}{28},\frac{9}{14}\right)$};

\end{axis}
\end{tikzpicture}
\caption{Comparison between the minimality constraint and the Plotkin-type bounds for $q=2$, $h=2$. The dashed black line is the classical Plotkin bound over $\F_4$, while the red line is the refined Plotkin bound for the non-faithful distribution $\nu_1=\nu_2=1/2$. The shaded red region is the range allowed simultaneously by minimality and the refined Plotkin bound.}
\label{fig:nonfaithful_refined_plotkin}
\end{figure}
\end{example}

\section{Conclusions and open questions}\label{sec:conclusion}

In this paper, we initiated the study of minimal additive codes from a finite-geometric point of view. The main idea is that, while linear codes are
described by projective systems of points, additive codes are described by projective systems of subspaces. This naturally leads to the notion of an additive strong blocking set, namely an $h$-projective system $\{\pi_1,\ldots,\pi_n\}$ whose union is a strong blocking set. We proved that this is the correct geometric object for minimal additive codes. More precisely, we established a one-to-one correspondence between equivalence classes of nondegenerate minimal additive codes and equivalence classes of additive strong blocking sets. This extends the classical
correspondence between minimal linear codes and strong blocking sets in the case in which $h=1$, and it also contains the theory of outer minimal codes as the special case in which the associated subspaces arise from field reduction.
We also studied several aspects of additive codes and minimal additive codes. We derived a total-weight formula depending on the vectorial dimensions of the coordinate subspaces, and obtained a refined Plotkin-type bound together with an asymptotic version. We applied these results to one-weight additive codes, obtaining faithful one-weight examples which are not equivalent to $\Fqh$-linear codes.
For minimal additive codes, we introduced the parameter $m(k,q,h)$ and derived existence results, geometric constructions, lower bounds, upper bounds, a minimum-distance bound, and asymptotic estimates.

Several natural questions remain open. We identify but a few.

\begin{enumerate}
    \item Determine sharper bounds for $m(k,q,h)$. Even for small values of $h$, it would be interesting to understand whether the shortest minimal additive codes arise from known constructions of strong
    blocking sets, or from new additive configurations.
    \item Construct small additive strong blocking sets which do not arise from
    field reduction.  Outer strong blocking sets give additive strong blocking sets whose subspaces belong to a Desarguesian spread. However, in the additive setting, the subspaces may have dimension smaller than $h-1$, or may have dimension $h-1$ without lying in a common Desarguesian     spread. Explicit families exploiting this flexibility could lead to shorter     minimal additive codes than those coming from the $\Fqh$-linear setting.
    \item Study one-weight additive codes and their role as minimal additive
    codes. Every one-weight additive code is minimal, and the examples constructed in
    this paper show that there exist faithful one-weight additive codes which
    are not equivalent to any $\Fqh$-linear code. It would be interesting to
    classify such examples, or at least to identify further geometric
    constructions of one-weight additive codes outside the linear setting.
    \item Investigate possible applications of minimal additive codes.
    Minimal linear codes have applications in secret sharing schemes, while
    additive codes naturally appear in quantum coding theory. It would be
    interesting to understand whether minimal additive codes, or particular
    families of additive strong blocking sets, have applications in these or
    related contexts.
\end{enumerate}

\bigskip

\bigskip

\bibliographystyle{alpha}
\bibliography{references.bib} 

@article{ball2020additive,
  title={On additive {MDS} codes over small fields},
  author={Ball, Simeon and Gamboa, Guillermo and Lavrauw, Michel},
  journal={Advances in Mathematics of Communications},
  volume={17},
  number={4},
  pages={828--844},
  year={2023},
  publisher={American Institute of Mathematical Sciences}
}

@article{jamison1977covering,
  title={Covering finite fields with cosets of subspaces},
  author={Jamison, Robert E},
  journal={Journal of Combinatorial Theory, Series A},
  volume={22},
  number={3},
  pages={253--266},
  year={1977},
  publisher={Elsevier}
}

@book{tsfasman2007algebraic,
  title={Algebraic Geometric Codes: Basic Notions},
  author={Tsfasman, Michael A. and Vl{\u{a}}duț, Serge G. and Nogin, Dmitry},
  volume={1},
  year={2007},
  publisher={American Mathematical Soc.}
}

@article{adriaensen2023additive,
  title={On additive {MDS} codes with linear projections},
  author={Adriaensen, Sam and Ball, Simeon},
  journal={Finite Fields and Their Applications},
  volume={91},
  pages={102255},
  year={2023},
  publisher={Elsevier}
}

@article{tang2021full,
  title={Full characterization of minimal linear codes as cutting blocking sets},
  author={Tang, Chunming and Qiu, Yan and Liao, Qunying and Zhou, Zhengchun},
  journal={IEEE Transactions on Information Theory},
  volume={67},
  number={6},
  pages={3690--3700},
  year={2021},
  publisher={IEEE}
}

@article{heger2021short,
  title={Short minimal codes and covering codes via strong blocking sets in projective spaces},
  author={H{\'e}ger, Tam{\'a}s and Nagy, Zolt{\'a}n L{\'o}r{\'a}nt},
  journal={IEEE Transactions on Information Theory},
  volume={68},
  number={2},
  pages={881--890},
  year={2021},
  publisher={IEEE}
}

@article{alfarano2022three,
  title={Three combinatorial perspectives on minimal codes},
  author={Alfarano, Gianira N. and Borello, Martino and Neri, Alessandro and Ravagnani, Alberto},
  journal={SIAM Journal on Discrete Mathematics},
  volume={36},
  number={1},
  pages={461--489},
  year={2022},
  publisher={SIAM}
}

@article{alfarano2022geometric,
  title={A geometric characterization of minimal codes and their asymptotic performance},
  author={Alfarano, Gianira N. and Borello, Martino and Neri, Alessandro},
  journal={Advances in Mathematics of Communications},
  volume={16},
  number={1},
  pages={115--133},
  year={2022},
  publisher={Advances in Mathematics of Communications}
}

@article{ball2025additive,
  title={Additive codes from linear codes},
  author={Ball, Simeon and Popatia, Tabriz},
  journal={IEEE Transactions on Information Theory},
  volume={72},
  number={1},
  pages={342--345},
  year={2025},
  publisher={IEEE}
}

@article{bartoli2025long,
  title={Long {QMDS} additive code},
  author={Bartoli, Daniele and Giannoni, Alessandro and Marino, Giuseppe and Zhou, Yue},
  journal={arXiv preprint arXiv:2509.03186},
  year={2025}
}

@article{bartoli2023small,
  title={Small strong blocking sets by concatenation},
  author={Bartoli, Daniele and Borello, Martino},
  journal={SIAM Journal on Discrete Mathematics},
  volume={37},
  number={1},
  pages={65--82},
  year={2023},
  publisher={SIAM}
}

@article{kurz2024additive,
  title={Additive codes attaining the {G}riesmer bound},
  author={Kurz, Sascha},
  journal={arXiv preprint arXiv:2412.14615},
  year={2024}
}

@article{alfarano2024outer,
  title={Outer strong blocking sets},
  author={Alfarano, Gianira N. and Borello, Martino and Neri, Alessandro},
  journal={The Electronic Journal of Combinatorics},
  year={2024}
}

@article{ball2025griesmer,
  title={Griesmer type bounds for additive codes over finite fields, integral and fractional {MDS} codes},
  author={Ball, Simeon and Lavrauw, Michel and Popatia, Tabriz},
  journal={Designs, Codes and Cryptography},
  volume={93},
  number={1},
  pages={175--196},
  year={2025},
  publisher={Springer US New York}
}

@article{alon2024strong,
  title={Strong blocking sets and minimal codes from expander graphs},
  author={Alon, Noga and Bishnoi, Anurag and Das, Shagnik and Neri, Alessandro},
  journal={Transactions of the American Mathematical Society},
  volume={377},
  number={08},
  pages={5389--5410},
  year={2024}
}

@article{dastbasteh2025polynomial,
  title={Polynomial representation of additive cyclic codes and new quantum codes},
  author={Dastbasteh, Reza and Shivji, Khalil},
  journal={Advances in Mathematics of Communications},
  volume={19},
  number={1},
  pages={49--68},
  year={2025},
  publisher={Advances in Mathematics of Communications}
}

@article{dastbasteh2024new,
  title={New quantum codes from self-dual codes over $\mathbb{F}_4$},
  author={Dastbasteh, Reza and Lison{\v{e}}k, Petr},
  journal={Designs, Codes and Cryptography},
  volume={92},
  number={3},
  pages={787--801},
  year={2024},
  publisher={Springer}
}

@article{d2026generalized,
  title={Generalized {H}amming weights of additive codes and geometric counterparts},
  author={D’haeseleer, Jozefien and Kurz, Sascha},
  journal={Designs, Codes and Cryptography},
  volume={94},
  number={7},
  pages={142},
  year={2026},
  publisher={Springer}
}

@article{brouwer1978blocking,
  title={The blocking number of an affine space},
  author={Brouwer, Andries and Schrijver, Alexander},
  journal={Journal of Combinatorial Theory Series A},
  volume={24},
  number={2},
  pages={251--253},
  year={1978},
  publisher={Elsevier}
}

@article{bishnoi2024blocking,
  title={Blocking sets, minimal codes and trifferent codes},
  author={Bishnoi, Anurag and D'haeseleer, Jozefien and Gijswijt, Dion and Potukuchi, Aditya},
  journal={Journal of the London Mathematical Society},
  volume={109},
  number={6},
  pages={e12938},
  year={2024},
  publisher={Wiley Online Library}
}

@article{grassl2021algebraic,
  title={Algebraic quantum codes: linking quantum mechanics and discrete mathematics},
  author={Grassl, Markus},
  journal={International Journal of Computer Mathematics: Computer Systems Theory},
  volume={6},
  number={4},
  pages={243--259},
  year={2021},
  publisher={Taylor \& Francis}
}

@inproceedings{massey1993minimal,
  title={Minimal codewords and secret sharing},
  author={Massey, James L.},
  booktitle={Proceedings of the 6th joint Swedish-Russian international workshop on information theory},
  pages={276--279},
  year={1993}
}

@article{Lavrauw2028Field,
  title={Field reduction and linear sets in finite geometry},
  author={Lavrauw, Michel and Van de Voorde, Geertrui},
  journal={Contemporary Mathematics},
  year={2015}
}

@article{fancsali2014lines,
  title={Lines in higgledy-piggledy arrangement},
  author={Fancsali, Szabolcs and Sziklai, P{\'e}ter},
  journal={The Electronic Journal of Combinatorics},
  pages={P2--56},
  year={2014}
}

\end{document}